\documentclass[a4paper,onecolumn,11pt,accepted=2023-06-02]{quantum_style}
\pdfoutput=1

\usepackage[margin=2.8cm]{geometry}
\usepackage{amsfonts}
\usepackage{graphicx}
\usepackage{caption}
\usepackage{thmtools,thm-restate}
\usepackage{amssymb}
\usepackage{amsmath}
\usepackage{amsthm}
\usepackage{mathrsfs}
\usepackage{bm}
\usepackage{gensymb}

\usepackage{url}

\usepackage{doi}

\usepackage{algorithm}
\usepackage{algpseudocode}
\usepackage[table,xcdraw]{xcolor}

\usepackage{xcolor}
\usepackage{setspace}
\usepackage{comment}

\usepackage{braket} 

\DeclareFontFamily{OT1}{pzc}{}
\DeclareFontShape{OT1}{pzc}{m}{it}{<-> s * [1.15] pzcmi7t}{}
\DeclareMathAlphabet{\mathpzc}{OT1}{pzc}{m}{it}

\newcommand{\N}{\mathbb{N}}
\newcommand{\R}{\mathbb{R}}

\newcommand{\fA}{\mathpzc{A}}

\newcommand{\fD}{\mathpzc{D}} 
\newcommand{\fE}{\mathpzc{E}}
\newcommand{\fG}{\mathpzc{G}}

\newcommand{\fJ}{\mathpzc{J}}
\newcommand{\fL}{\mathpzc{L}}
\newcommand{\fM}{\mathpzc{M}}

\newcommand{\fV}{\mathpzc{V}}

\newcommand{\fd}{\mathpzc{d}}
\newcommand{\fm}{\mathpzc{m}}
\newcommand{\fn}{\mathpzc{n}}
\newcommand{\fs}{\mathpzc{s}}

\newcommand{\mI}{\mathcal{I}}
\newcommand{\mM}{\mathcal{M}}
\newcommand{\mO}{\mathcal{O}}
\newcommand{\mP}{\mathcal{P}}

\newcommand{\Tk}{T_{k\text{-}\mathrm{means}}}

\newcommand{\poly}{\mathop{poly}}

\newcommand{\Ln}{L_{\n}}
\newcommand{\Lnt}{\tilde{L}_{\n}}
\newcommand{\one}{\mathbf{1}}
\renewcommand{\vec}[1]{\mathbf{#1}}

\DeclareMathOperator*{\argmin}{arg\,min}
\DeclareMathOperator{\cut}{cut}

\DeclareMathOperator{\RC}{RatioCut}
\DeclareMathOperator{\vol}{vol}
\DeclareMathOperator{\n}{norm}

\newtheorem{theorem}{Theorem}

\newtheorem{lemma}{Lemma}
\newtheorem{corollary}{Corollary}

\newcommand\numberthis{\addtocounter{equation}{1}\tag{\theequation}}

\begin{document}

\setstretch{1}

\title{Quantum Motif Clustering}

\author[1]{Chris Cade}

\author[1]{Farrokh Labib}

\author[1]{Ido Niesen}

\affil[1]{QuSoft \& CWI, Amsterdam, the Netherlands.}
\date{}

\captionsetup{width=400pt}

\maketitle

\begin{abstract}
\noindent We present three quantum algorithms for clustering graphs based on higher-order patterns, known as motif clustering. One uses a straightforward application of Grover search, the other two make use of quantum approximate counting, and all of them obtain square-root like speedups over the fastest classical algorithms in various settings. In order to use approximate counting in the context of clustering, we show that for general weighted graphs the performance of spectral clustering is mostly left unchanged by the presence of constant (relative) errors on the edge weights. Finally, we extend the original analysis of motif clustering in order to better understand the role of multiple `anchor nodes' in motifs and the types of relationships that this method of clustering can and cannot capture.
\end{abstract}

\section{Introduction}

The study of complex networks has impacted many fields of science~\cite{strogatz2001exploring}, including biology~\cite{albert2005scale, shen2002network}, sociology~\cite{wasserman1994social}, neuroscience~\cite{bassett2017network}, and finance~\cite{acemoglu2015systemic, gai2010contagion}. In particular, it is commonplace to study the connectivity patterns of networks at the edge and vertex level in order to uncover important structures in the underlying data. One method that provides insight into the connectivity structure of a network is graph clustering, which entails finding groups of highly connected vertices in order to uncover underlying community structures. There are many efficient (heuristic) algorithms for graph clustering, including the theoretically well-motivated \emph{k-means spectral clustering}\footnote{See~\cite{luxburg2007tutorial} for a historical overview and list of references.}. Here, given an integer $k$, the eigenvectors corresponding to the smallest $k$ eigenvalues of the graph Laplacian are used as a feature set for a $k$-means clustering algorithm. It has been shown that in certain circumstances spectral clustering leads to the discovery of \emph{optimal} graph partitions~\cite{peng2015partitioning}. 

Recently, it is becoming popular to study more sophisticated connectivity patterns. This can be done in the context of, for example, hypergraphs\footnote{\href{https://www.quantamagazine.org/how-big-data-carried-graph-theory-into-new-dimensions-20210819/}{www.quantamagazine.org/how-big-data-carried-graph-theory-into-new-dimensions-20210819/}} that can express multiple-vertex relationships, or via small subgraphs, also known as \emph{motifs}, which can be used to study higher-order connectivity patterns between vertices. The latter has become a useful tool for providing deeper insight into a network's function and structure, although often the detection of these motifs remains computationally challenging~\cite{masoudi2012building}.

In \cite{benson16}, Benson et al.~propose an algorithm for clustering a graph based on its \emph{motif connectivity}. Their algorithm, which makes use of spectral clustering and therefore comes with theoretical guarantees~\cite{peng2015partitioning}, can be used to uncover collections of vertices that are highly connected via particular motifs, rather than just by edges as in the ordinary case. The authors apply their technique to the well known \emph{C.~elegans} neuronal network and to a transportation reachability network, with a particular motif used for each case, and find that the motif clustering reveals network organisation not made apparent by clustering through edge-based connectivity alone. 

\

As a motivating (toy) example, consider the graph shown in the middle of Figure~\ref{fig:motif_clustering}, which could for example represent a financial transaction network: vertices corresponding to financial entities, and unweighted edges between them denoting transactions (for example, of a value beyond a particular threshold). Consider also the motif shown at the top of Figure~\ref{fig:motif_clustering}, which represents the situation wherein two entities, given by the \emph{anchor nodes} shown in dark red, trade indirectly through three intermediate entities according with edge pattern of the motif. Suppose that we are interested in clustering the nodes of the graph into groups that don't trade with each other directly, but instead do so only by means of intermediate nodes in accordance to the structure given by the motif. The method of motif clustering achieves precisely such a clustering.

As shown by~\cite{benson16}, obtaining a motif clustering of the original graph can be done in two steps. First, we construct the \emph{motif graph}, displayed in the bottom of Figure~\ref{fig:motif_clustering}. This graph has the same vertex set as the original graph, but a different edge set: any instance of the motif in the original graph corresponds to an edge in the motif graph connecting the two anchor nodes of the motif instance in question. All edges of the motif graph have integer weights that correspond to the number of motif instances in the original graph that have the two edge endpoints as anchor nodes. In the second step we use $k$-means spectral clustering on the motif graph to obtain the required motif clustering of the original graph.


\begin{figure}[htb]
    \centering
    \includegraphics[scale=0.6]{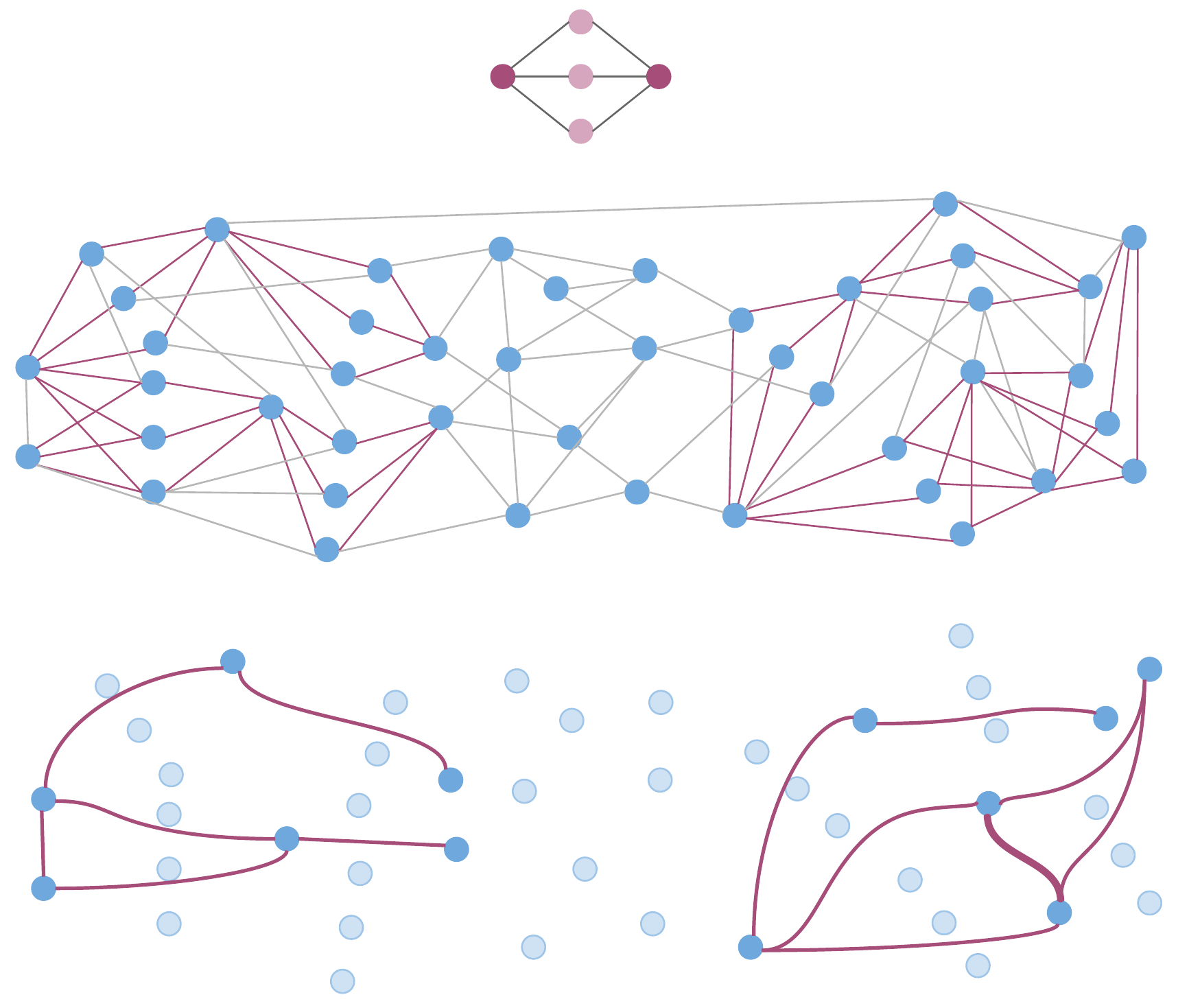}
    \caption{Example of motif clustering. \emph{Top}: The motif, with anchor nodes in dark red. \emph{Middle}: An example graph, instances of the motif are represented by red edges, other edges are grey. \emph{Bottom}: The motif graph. All edges have weight 1 except for a single weight-2 edge (bold). The motif clusters are clear: one on the left, one on the right.}
    \label{fig:motif_clustering}
\end{figure}

\ 

In this paper we present three quantum algorithms that can perform motif clustering faster than the classical algorithm presented in~\cite{benson16}. The majority of the quantum speedup comes from faster finding and/or counting of motifs in the graph, a task that is often computationally demanding. Our speedups are of the Grover variety: at most quadratic, and sometimes less, depending on the choice of motif and the sparsity of the graph. Our reason for presenting several quantum algorithms is that we have a choice between using Grover search or quantum approximate counting, as well as the option of constructing the motif graph in its entirety, or only giving query access to it. Depending on the input graph, one will be favorable over the other. We should add that, as argued in~\cite{babbush2021focus}, for quantum algorithms based on Grover-type speedups to become practically interesting, substantial improvement in qubit counts, physical gate errors and/or error correction schemes are required.

We also prove some technical lemmas related to performing spectral clustering in the presence of errors on the weights of the edges of a graph -- this is the case in our application when the number of motifs connecting two vertices is estimated using quantum approximate counting -- which might be of independent interest. Along the way, we give a simple `no-go' argument to show that spectral clustering with errors on \emph{normalized} Laplacians does not come with the same guarantees as in the unnormalized case, but nevertheless, numerical experiments suggest that it still performs well in practice. An interesting open question is whether we can turn this empirical observation into a theoretical one. 

Finally, in the Appendix we discuss the role of anchor nodes in motifs, extending the analysis of Benson et al. More specifically, we argue that motif clustering should be used only for two-anchor node motifs, which express pairwise relationships between vertices. If, instead, we want to cluster using relationships between more than two vertices -- which is what motifs with more than two anchor nodes attempt to capture -- we should do so within the context of hypergraphs rather than that of motif clustering. 

\paragraph{Organisation}
After introducing some notation in Section~\ref{sec:prelims}, we begin by explaining the concept of motif clustering and discuss previous work in the (classical) literature in Section~\ref{sec:motif_clustering}. Following this, we summarise our main results in Section~\ref{sec:results}. Section~\ref{sec:classical-mc} describes a classical algorithm for motif clustering, and introduces the notation that we use throughout the rest of the paper. In Section~\ref{sec:quantum_algorithms} we introduce our quantum tools and use those to construct our three quantum algorithms: one based on Grover search, the other two on quantum approximate counting. Finally, in Section~\ref{sec:approx-adjacency-matrix} we consider the effect of quantum approximate counting, both analytically and numerically, on the guarantees that come with spectral clustering. In the Appendix we discuss in detail the role of anchor nodes in motifs.

\section{Preliminaries}
\label{sec:prelims}

Before we discuss the concept of motif clustering, we first introduce some notation. In addition, when comparing how well our quantum algorithms perform relative to their classical counterparts, we want to be able to talk about their run-times. In the section below, we make precise what we mean with run-time.

\subsection{Notation}
For an integer $k\geq 1$, we write $[k]:=\{1,\dots, k\}$. For a set $W$, we denote its size by $|W|$. We write $G=(V,E)$ for a directed graph with vertex set $V$ and edge set $E$, where $n = |V|$ denotes the number of vertices, and $m = |E|$ the number of edges, and assume a fixed ordering of the vertices in $V$ that allows for a natural identification of $V$ with $[n]$. 

For $v\in V$, we define $d_v$ to be the degree of $v$ and $d := \max_{v\in V}d_v$ the maximum degree of any vertex in the graph. We use $A$ to denote the adjacency matrix of the graph, and for each vertex $v$ assume that we know $d_v$ and that we have query access to the weighted adjacency list $J_v : [d_v] \rightarrow [n] \times \mathbb{R}_{\geq 0}$,  which is a function that assigns labels and weights to the neighbours of $v$. We will call such access `adjacency list access' to $G$. For a subset $W\subset V$, we write $\overline{W}$ for the complement, i.e. $\overline{W} = V\setminus W$. A $k$-\emph{partition} of $V$ is a collection of pairwise disjoint subsets $W_1,\dots, W_k\subset V$ such that $V=\cup_{i\in [k]}W_i$.

We will often consider the Laplacian $L$ of an $n$-vertex graph $G$, and its normalized equivalent, the normalized Laplacian $\tilde{L}$. Let $D$ be the diagonal $n \times n$ matrix of (weighted) vertex degrees in $G$, and $A$ the adjacency matrix. Then the Laplacian is defined as $L := D - A$, and the normalized Laplacian as $\tilde{L} := D^{-\frac12}LD^{-\frac12} = I - D^{-\frac12}AD^{-\frac12}$ (which is only defined for graphs for which every vertex has a positive degree).

Finally, given any $n \times n$ real symmetric matrix, such as the graph Laplacian, we assume that the eigenvalues $\lambda_1 \leq \lambda_2 \leq \cdots \leq \lambda_n$ are ordered by increasing value and denote by $\vec{v}_1, \vec{v}_2, \dots, \vec{v}_n$ the corresponding eigenvectors, which we assume to be normalized. In particular, when we mention the `first $k$ eigenvectors' of a matrix, we are referring to the eigenvectors corresponding to the smallest $k$ eigenvalues. For graph Laplacians, which are positive semi-definite, the smallest eigenvalues will be those closest to zero.

\subsection{Query and time complexity}
All our quantum algorithms assume coherent access to the input graph in the form of quantum queries to the adjacency lists. More explicitly, given the maps $\{J_v: v\in V\}$ introduced above, coherent access to the adjacency lists means that we have access to the following unitary
\begin{equation}
    \ket{v}\ket{i} \ket{0}  \mapsto \ket{v}\ket{i} \ket{J_v(i)}\, ,
    \label{eq:QROM}
\end{equation}
and its inverse, where $v \in V$, $i\in [d]$ and $\ket{J_v(i)}$ contains two registers, one for the label of $i$-th neighbor of $v$, and one for the weight of the edge. When we talk about the \emph{query complexity} of an algorithm, we mean the number of times the algorithm applies the unitary in Eq.~\eqref{eq:QROM}. If the adjacency lists are sorted (according to some ordering of the vertices in $V$), then the only type of access to the input graph our algorithms require is this\footnote{Except when we also make use of the quantum graph sparsification algorithm by Apers and de Wolf~\cite{apers20}, which requires QRAM -- see Section~\ref{sec:results} for details.}.

If we are not provided with coherent access to the adjacency lists, or they are not sorted, then we must provide our own (perhaps sorted) coherent access, which will require classically writing to a QRAM, using at most $O(nd)$ operations. We note that since the run-times of our quantum algorithms are larger than $O(nd)$, our speedups persist even if we pay this cost up front. 

Finally, by the \emph{run-time} of a quantum algorithm we mean the total number of elementary gates, QRAM writes, and queries made by the algorithm. Our definition of run-time is the same as that of Apers and de Wolf~\cite{apers20}.


\section{Motif clustering}\label{sec:motif_clustering}

The idea behind motif clustering is to partition the vertices of a graph into several clusters based on a higher-order structural pattern, called a \emph{motif}. The partitions obtained through motif clustering should be such that any two vertices within a particular cluster are part of relatively many connected occurrences of the motif in the graph, whereas two vertices in different clusters should participate in relatively few connected motif occurrences. This statement will be made precise below. An example of motif clustering for a particular motif is shown in Figure~\ref{fig:motif_clustering}.

In this section, which is based on the content of Benson et al.~\cite{benson16}, we set the stage by introducing the reader to the necessary concepts and definitions that we use throughout the paper.

\subsection{Graph motifs}
\label{sec:graph_motifs}

A motif $M = (V_M, E_M, V_A)$ of size $s$ is a connected, unweighted graph with $s$-sized vertex set $V_M$, and edge set $E_M$. Throughout this paper we assume that $s$ is a constant. The motif comes with a set of \emph{anchor nodes} $V_A \subseteq V_M$, which will become relevant when we discuss motif cuts. 

Given a particular motif $M$ and an unweighted graph $G=(V,E)$, we will be interested in occurrences of the motif $M$ in $G$, which can be \emph{functional} or \emph{structural}~\cite{sporns2004motifs}. Formally, a \emph{motif assignment} is an injective map $\iota: V_M \rightarrow V$. For functional motifs, we require that  $(\iota(v), \iota(w)) \in E$ if $(v, w) \in E_M$ for every $v,w \in V_M$. That is, any two vertices in $\iota(V_M)$ should have an edge in $G$ whenever the corresponding vertices in the motif have one, but there can be additional edges in $G$ not present in the motif itself. For structural motifs, we have that $(\iota(v), \iota(w)) \in E$ if \emph{and only if} $(v, w) \in E_M$ for all $v,w \in V_M$, and therefore the motif is graph-isomorphic to $\iota(V_M)$ (i.e.~both the edges \emph{and} non-edges coincide).

Because we are interested in the motif occurring in $G$ as a pattern irrespective of the actual vertex assignment given by the mapping $\iota$, we next define an equivalence relation on the set of motif assignments. Two motif assignments $\iota$ and $\iota'$ are considered equivalent if $\iota(V_M) = \iota'(V_M)$ as sets and also $\iota(V_A) = \iota'(V_A)$ as sets. A single equivalence class is called a \emph{motif instance}. We write $\mI = \mI(G,M)$ for the set of all motif instances of $M$ in $G$. Moreover, we say that a motif instance has a vertex $u\in V$ as an anchor node if, for any assignment $\iota$ in the equivalence class of the instance, $u \in \iota(V_A)$. Note that this definition is well-defined, as it does not depend on the choice of assignment $\iota$. In Appendix~\ref{sec:motif_isomorphisms}, we further elaborate on when two motif assignments are equivalent, and how this equivalence is related to symmetries of the motif itself. 

Note that, motif clustering can be applied to both directed and undirected (but unweighted) graphs. When the graph is (un)directed, the motif should also be (un)directed for motif instances to exist within the graph.

Following~\cite{benson16}, we focus on structural motifs in this work. Note that a functional motif can be thought of as a combination of several structural motifs\footnote{A functional motif only specifies the edges that are present in the motif. For all other (unspecified) edges, we define a structural motif for every possible assignment of either edge or non-edge to the set of edges unspecified by the functional motif.}. Since the framework introduced by Benson et al.~\cite{benson16} can be extended to consider several motifs simultaneously -- see Appendix~\ref{app:multiple_motifs}, and the same extensions work for our framework, we also capture the case of functional motifs. 


\subsection{Motif cuts}
\label{sec:motif_cuts}

A common method for clustering graphs is to find minimal normalized cuts. This corresponds to finding a partition of the vertex set that minimizes the total weight of the cuts (edges connecting different partitions) whilst maximizing the volumes or sizes of the partitions. The number of partitions (often denoted $k$) is fixed in advance. Motif clustering works analogously, in the sense that it minimizes the number of \emph{motif cuts} whilst maximizing the \emph{motif volume} or partition size.


More formally, let $W\subset V$ be a subset of vertices in the graph. An ordinary graph cut with respect to $W$ and its complement $\bar{W}$ is given by the number of edges that have one endpoint in $W$ and the other in $\bar{W}$. Similarly, we \emph{could} define a motif cut with respect to $W$ and $\bar{W}$ to be the number of motif instances that have one or more vertices in $W$ and one or more vertices in $\bar{W}$. However, following~\cite{benson16}, we want to incorporate the idea that each individual motif instance signifies a mutual relationship between a specific subset of vertices of the motif, called \emph{anchor nodes}. As such, given a motif $M$, a \emph{motif cut} in $G$ with respect to $W$ and its complement $\bar{W}$, denoted by $\text{cut}_{(G,M)}(W)$, is defined to be the number of motif instances that have at least one anchor node in both $W$ and $\bar{W}$: 
\begin{equation}
    \text{cut}_{(G,M)}(W) := \left| \{ \iota \in \mI | (W \cap \iota(V_A) \neq \emptyset) \land  (\bar{W} \cap \iota(V_A) \neq \emptyset)\} \right|\,.
    \label{eq:motif-cond-partition}
\end{equation}
Moreover, analogous to the ordinary volume given by the sum of all degrees of vertices in a given set, define the \emph{motif volume} $\text{vol}_{(G,M)}(W)$ of $W$ to be the number of anchor nodes in motif instances that appear in $W$:
\begin{equation}
    \vol_{(G,M)}(W) := \sum_{\iota \in \mI}\left|W \cap \iota(V_A)\right|\,.
    \label{eq:motif-rcut-partition}
\end{equation}

Given the notions of motif cut and motif volume, the \emph{motif conductance} and the \emph{motif ratio cut} of the set $W$ are given by\footnote{Graph conductance is usually defined with $\text{min} \left(\text{vol}_{(G,M)}(W), \text{vol}_{(G,M)}(\bar{W}) \right)$ in the denominator. However, when considering more general partitions of the graph in possibly more than two sets, such as arbitrary $k$-partitions, we divide by $\text{vol}_{(G,M)}(W)$ instead: see definitions of Ncut in~\cite{luxburg2007tutorial} and the related $k$-way expansion constant in~\cite{peng2015partitioning}.}
\begin{equation}
        \phi_{(G,M)}(W) := \frac{\text{cut}_{(G,M)}(W)}{\text{vol}_{(G,M)}(W)} \qquad \text{and} \qquad \text{RatioCut}_{(G,M)}(W) := \frac{\text{cut}_{(G,M)}(W)}{|W|}
    \label{eq:motif-cond} 
\end{equation}
respectively, where we define $\phi_{(G,M)}(W) = 0$ and $\text{RatioCut}_{(G,M)}(W) = 0$ when $\text{cut}_{(G,M)}(W) = 0$. 

\ 

\noindent Given some integer $k \in \N$ chosen in advance, we can also define the conductance and ratio cut relative to a $k$-partition $W_1, \ldots, W_k$. The \emph{motif conductance} and \emph{motif ratio cut} of the partition $\{W_i\}_{i=1}^k$ of $V$ are given by
\begin{eqnarray*}
    \phi_{(G,M)}(W_1, \ldots, W_k) := \sum_{i=1}^k \phi_{(G,M)}(W_i) \qquad \text{and} \\
    \RC_{(G,M)}(W_1, \ldots, W_k) :=  \sum_{i=1}^k \RC_{(G,M)}(W_i)
\end{eqnarray*}
respectively.

The goal of motif spectral clustering, for a fixed $k \in \N$ chosen in advance, is to find a partition $\{W_i\}_{i=1}^k$ that minimizes $\phi_{(G,M)}(W_1, \ldots, W_k)$ or $\RC_{(G,M)}(W_1, \ldots, W_k)$: this will result in a partitioning of the graph into ($k$) clusters of vertices that are highly connected via the target motif, whilst very few motifs connect vertices in different clusters. It turns out that, for motifs with two or three anchor nodes, one can translate the two minimization problems above to the problems of ordinary conductance or ratio cut minimization of an auxiliary, weighted graph, called the \emph{motif graph}~\cite{benson16} which we introduce in the next section.

\subsection{The motif graph}
\label{sec:motif-graph}

Benson et al.~\cite{benson16} show that for motifs with two or three anchor nodes, minimizing $\phi_{(G,M)}(W_1, \ldots, W_k)$ is equivalent to minimizing the ordinary conductance on a weighted graph $\fG = (V, \fE, \fA)$ that can be constructed from the graph $G$, which we term the \emph{motif graph} of $G$ given motif $M$. The graph $\fG$ has the same vertex set\footnote{For practical details regarding entirely disconnected vertices aside, see Section~\ref{sec:disconnected_vertices}.} as $G$, but in general a different set of edges, which are now integer weighted. Also, whereas both $G$ and $M$ can be directed, $\fG$ is always an undirected graph. For notation we will use ordinary characters $G$, $V$, $E$, etc.~when referring to the original graph $G$, and calligraphic characters $\fG$, $\fE$, etc.~to refer to the motif graph $\fG$.

Given $G$ and $M$, the edge set $\fE$ of $\fG$ is given by $\fE = \{(v,w) \in V \times V | \exists\, \iota \in \mI: \{v,w\} \subseteq \iota(V_A) \}$, i.e. two vertices $u,v \in V$ are connected by an edge in $\fG$ if they are both anchor nodes of a motif instance $\iota \in \mI$ of $G$. The motif weighted adjacency matrix $\fA$ of $\fG$ has integer coefficients given by
\[
    \fA_{uv} = \left|\{ \iota \in \mI \mid \{u,v\} \subseteq \iota(V_A) \}\right|\,,
\]
i.e. $\fA_{uv}$ is equal to the number of motif instances in $G$ that contain both $u$ and $v$ as anchor nodes. For $u\in V$, we define the \emph{motif degree} $\fd_u$ of $u$ to be the total number of edges in $\fG$ connected to $u$,
\[
    \fd_u := |\{v \in V: \fA_{uv} > 0 \}|,
\]
and the \emph{motif strength} $\fs_u$ to be the sum of all weights of edges connected to $u$:
\[
    \fs_u := \sum_{v\in V} \fA_{uv}.
\]

\noindent Given $W \subset V$, $W \neq \emptyset$ let  
\[
    \vol_{\fG}(W) := \sum_{u\in W} \fs_u
\]
be the volume of $W$ in $\fG$, and 
\[
    \cut_{\fG}(W) := \sum_{u\in W, v \in \bar{W}} \fA_{uv}
\]
the cut induced by $W$ in $\fG$. The conductance of and ratio cut of $W$ in $\fG$ are given by 
\[
    \phi_{\fG}(W) := \frac{\text{cut}_{\fG}(W)}{\vol_{\fG}(W)} \qquad \text{and} \qquad \RC_{\fG}(W) := \frac{\cut_{\fG}(W)}{|W|}\,,
\]
respectively. Finally, for a $k$-partition $W_1,\ldots, W_k$ of $V$ with $W_k \neq \emptyset \forall k$, we define
\[
    \phi_{\fG}(W_1, \ldots, W_k) := \sum_{i=1}^k \phi_{\fG} (W_i),
\]
and
\[
    \RC_{\fG}(W_1, \ldots, W_k) := \sum_{i=1}^k \RC_{\fG} (W_i)
\]
for the conductance and ratio cut of the partition in $\fG$.

\ 

\noindent Benson et al. prove the following results relating motif conductances and volumes in $G$ to ordinary conductances and volumes in $\fG$:

\begin{lemma}[\cite{benson16}]\label{lem:benson1}
Let $G = (V,E)$ be an unweighted graph, $M = (V_M, E_M, V_A)$ a motif with  $ |V_A| \geq 2$ anchor nodes\footnote{Note that one-anchor node motifs cannot be used for clustering, so the assumption that $|V_A| \geq 2$ is no restriction in practice.}, $\fG$ the motif graph constructed from $G$ and $M$, and $W \subset V$. Then
\[
    \vol_{(G,M)}(W) = \frac{1}{|V_A|-1} \vol_{\fG}(W).
\]
Moreover, if  $|V_A|\in\{2,3\}$, then
\[
    \cut_{(G,M)}(W) = c \, \cut_{\fG}(W),
\]
where the constant $c = 1$ if $|V_A| = 2$, and $c = \frac{1}{2}$ if $|V_A| = 3$.
\end{lemma}
The proof for $|V_A| = 2$ is not given in~\cite{benson16}, but it can be proven in exactly the same way as the proof for the $|V_A|=3$ case, see Appendix~\ref{app:motif_graph_cuts} for details. Why the second equation above only holds for motifs with two or three anchor nodes is discussed extensively in Appendix~\ref{sec:higher_order_motifs}.

For $k\in \N$, write $\mP_k(V)$ for the set of all $k$-partitions of $V$. The following corollary is immediate from Lemma~\ref{lem:benson1}.
\begin{corollary}[\cite{benson16}]
Let $G = (V,E)$ be an unweighted graph, $M = (V_M, E_M, V_A)$ a motif with $|V_A| \in \{2, 3\}$ anchor nodes and $\fG$ the motif graph constructed from $G$ and $M$. Then, for every subset $W \subset V$, we have
\[
    \phi_{(G,M)}(W) = \phi_{\fG} (W),
\]
and
\[
    \RC_{(G,M)}(W) = c_R \, \RC_{\fG}(W),
\]
where the constant $c_R = 1$ if $|V_A| = 2$, and $c_R = \frac{1}{2}$ if $|V_A| = 3$. Consequently, 
\[
    \argmin_{\{W_i\}_{i=1}^k \in \mP_k(V)}\phi_{(G,M)}(W_1, \ldots, W_k) =  \argmin_{\{W_i\}_{i=1}^k \in \mP_k(V)} \phi_{\fG}(W_1, \ldots, W_k),
\]
and
\[
    \argmin_{\{W_i\}_{i=1}^k \in \mP_k(V)}\RC_{(G,M)}(W_1, \ldots, W_k) =  \argmin_{\{W_i\}_{i=1}^k \in \mP_k(V)} \RC_{\fG}(W_1, \ldots, W_k).
\]
\end{corollary}
\noindent In particular, for motifs with two or three anchor nodes, in order to find partitions of $V$ that have small motif conductance or small motif ratio cut in $G$, we can instead solve the equivalent problem of finding partitions of $V$ that have small ordinary conductance or ratio cut in $\fG$ respectively. 

In order to obtain a partition that (approximately) minimizes the conductance or the ratio cut of $\fG$,~\cite{benson16} uses $k$-means spectral clustering on the motif graph $\fG$. We will describe how to do this in detail in Section~\ref{sec:classical-mc}. Thereafter, we will discuss how we can improve the classical algorithm using quantum algorithmic methods in Section~\ref{sec:quantum_motif_clustering}. We begin by first stating our results in Section~\ref{sec:results}.

\section{Results}\label{sec:results}

Searching for a motif in a graph is essentially an unstructured search problem. As such, we can speed up the parts of the classical algorithm that construct the motif graph by applying either Grover search or quantum counting. Which approach will be faster will depend on the properties of the input graph, and also affect what type of clustering we can employ. Using Grover search we can construct the motif graph \textit{exactly}, and then apply spectral clustering to its unnormalized \textit{or} normalized Laplacian while having guarantees on its behaviour. Using quantum approximate counting, on the other hand, can be faster than using Grover search but only approximately constructs the motif graph. As we show in Section~\ref{sec:approx-adjacency-matrix}, in this case we only have guaranteed behaviour when spectral clustering is applied to the \textit{unnormalized} Laplacian. Moreover, rather than pre-compute the entire motif graph $\fG$, for instance by explicitly writing down the motif adjacency matrix or motif adjacency lists, in some cases it can be more efficient to provide query access to it via some subroutine.


Taking the above considerations into account, we present below three quantum algorithms for motif spectral clustering that give speedups over the best classical algorithm in various situations. The complexities of the (quantum and classical) algorithms considered in this work are dominated by the time required to compute the edges and weights of the motif graph. 
The first two of our quantum algorithms focus on constructing the entire motif graph before applying some classical or quantum algorithm for spectral clustering; the third provides query access to the motif graph's adjacency lists, and then uses a fast quantum spectral clustering algorithm based on quantum graph sparsification by Apers and de Wolf~\cite{apers20}.

\paragraph{Input graph}
Let $G = (V,E)$ be the input graph that we want to cluster according to some $s$-vertex motif $M$. Since $s=2$ corresponds to an edge, without loss of generality we can and will assume throughout the rest of the paper that $s \geq 3$. We write $n = |V|$ for the number of vertices of $G$ and $d$ for the maximum degree (which can be $n$). We assume that we have adjacency list access to $G$. In addition, we will only consider constant-sized motifs, meaning that $s$ is independent of $n$. 

For the first quantum algorithm presented below, the run-time depends on whether the adjacency lists of the input graph are sorted (according to some chosen ordering of all vertices in $V$) or not.\footnote{The assumption of sorted adjacency lists is natural in the context of, for example, transaction networks, which are initialized as empty lists, and kept sorted as new transactions are added to the graph.} For the other two, we can sort the adjacency lists of the input graph beforehand without affecting their (asymptotic) run-times, and then use the algorithms as described in Section~\ref{sec:quantum_algorithms}, which assume that access to the input graph is provided via sorted adjacency lists.

Our quantum algorithms require coherent access to the input graph. If, rather than coherent access, we are given classical access to the input graph, we will first need to load the input graph into QRAM in time $O(nd)$. In addition, sorting the adjacency lists takes time\footnote{The $\tilde{O}$ notation hides poly-logarithmic factors in $n$.} $\tilde{O}(nd)$. If we either have to sort, load to QRAM, or both, we say that we need to \emph{pre-process} the input graph. In Appendix~\ref{app:adjacency_lists} we discuss how pre-processing the input graph affects the run-times of our algorithms.

\subsection{Algorithms and their run-times}

Below we provide the complexities for constructing the motif graph and obtaining the $k$ eigenvectors corresponding to the smallest $k$ eigenvalues of the motif graph Laplacian, where the latter are then used as input to $k$-means clustering. $k$-means clustering itself is a heuristic algorithm, with an exponential upper bound to its run-time, but in practice is usually significantly faster than this. For `well-clusterable' graphs, the $k$-means part of the algorithm can be done in nearly-linear time~\cite{peng2015partitioning}. For simplicity we will denote the time it takes to run $k$-means by $\Tk$, and note that generally this will not be the most expensive part of the algorithms.

\paragraph{Classical}
The classical algorithm of~\cite{benson16}, which can also be used in combination with the spectral estimation algorithm of~\cite{spielman2014nearly}, takes time $\tilde{O}(nd^{s-1})$ to obtain the $k$-smallest eigenvectors of the (optionally normalized) motif Laplacian, adding up to a total run-time of 
\[
    \tilde{O}(nd^{s-1} + \Tk)\,
\]
for the entire $k$-means motif spectral clustering algorithm. This is essentially\footnote{Except for certain specific motifs that can be counted more efficiently in some settings -- see Section~\ref{sec:constructing_motif_graph} for some examples.} optimal for any algorithm that makes use of the motif Laplacian, since to construct the motif graph exactly one needs to have counted all motif instances in the graph, which can be as large as $nd^{s-1}$, and hence counting these classically requires $\Omega(nd^{s-1})$ queries to the input graph via standard lower bounds on the query complexity of counting.

\paragraph{Quantum via Grover search} 
Our first quantum algorithm uses Grover search plus classical subroutines to find all motif instances in the graph and compute the weights of the edges in the motif graph \emph{exactly}, before applying the classical spectral clustering based on the spectral estimation algorithm of~\cite{spielman2014nearly}. 

\begin{theorem}[Motif clustering via Grover search]\label{theo:first_result}
Given a graph $G$ with maximum degree $d$ and a motif $M$ of size $s$, there exist quantum algorithms for \emph{exact} motif clustering under the following conditions and with the following run-times:
\begin{enumerate}
    \item If we do not have coherent access to the input graph, then there is a quantum algorithm for motif clustering with expected run-time
    \begin{equation}
        \tilde{O}(nd + \sqrt{nd^{s-1}\fM} + \Tk)\, ,
        \label{eq:Grover-run-time-pre-process}
    \end{equation}
    where $\fM$ is the total number of motif instances in the graph $G$. 
    
    \item If we have coherent access to the input graph, and the adjacency lists are sorted, then there is an algorithm that takes time
    \[
        \tilde{O}(\sqrt{nd^{s-1}\fM} + \Tk)\, 
    \]
    in expectation.
    
    \item If the adjacency lists are not sorted, then they can be pre-sortedto yield a quantum algorithm with the same run-time as given in Eq.~\eqref{eq:Grover-run-time-pre-process}. Otherwise, there is a quantum algorithm with expected run-time 
    \[
        \tilde{O}(\sqrt{nd^{s}\fM} + \Tk)\, .
    \]
\end{enumerate}
\end{theorem}
The run-times of Algorithms 1 and 2 above are analysed in Section~\ref{sec:motif_grover}. Given coherent access to the input graph, which of Algorithms 2 or 3 to use depends on $\fM$. The second is faster in case there are relatively few motif instances $\fM$ in total, i.e. when $\fM = o\left(\frac{n}{d^{s-2}}\right)$. If we lack any knowledge of a non-trivial bound on $\fM$, the sensible choice is to first sort all adjacency lists, since the algorithm so obtained is never slower than its classical counterpart.

In general, the best upper bound we can put on $\fM$ is $nd^{s-1}$, in which case the quantum algorithm runs in time $\tilde{O}(nd^{s-1})$ -- no better than classical. Hence, this algorithm provides a speedup whenever $\fM = o(nd^{s-1})$; we later show that for scale-free networks, which occur often in practice, this is indeed the case. The advantage of this algorithm is that by constructing the motif graph exactly, it can be used to perform spectral clustering using the eigenvectors of the ordinary as well as the \emph{normalized} Laplacian --- see Section~\ref{sec:classical-mc} below for a detailed discussion on the difference between clustering with the Laplacian or its normalized counterpart.

\paragraph{Quantum via approximate counting and classical spectral clustering} 
Our second quantum algorithm uses quantum approximate counting to estimate the weights of the motif graph, followed by a classical spectral clustering routine. To perform spectral clustering using the eigenvectors of the \emph{unnormalized} Laplacian, it is sufficient to approximate the entries of the motif adjacency matrix up to constant multiplicative error, and then use the spectral clustering algorithm based on the spectral estimation algorithm of~\cite{spielman2014nearly}. We can also perform spectral clustering using the \emph{normalized} Laplacian, but in this case we lack the theoretical guarantees present in the unnormalized case. However, our numerical simulations suggest that spectral clustering with the normalized Laplacian on the approximate motif graph \emph{does} actually work in practice -- see Section~\ref{sec:quantum_algorithms} for details. We prove the following result in Section~\ref{sec:motif_clustering_quantum_counting}.

\begin{theorem}[Motif clustering via quantum counting \& classical clustering]\label{theo:second_result}
Given a graph $G$ with maximum degree $d$ and a motif $M$ of size $s$, there exists a quantum algorithm for \textit{approximate} motif clustering that takes time 
\[
    \tilde{O}(nd^{l + \frac{s}{2} - 1} + \Tk)\,
\]
in expectation, where $l$ is the maximum distance between any two anchor nodes in the motif.
\end{theorem}
Hence, we obtain a speedup over the classical algorithm whenever $l < s/2$. As a simple example, consider the case of a triangle motif, so that $s=3$ and $l=1$, and where the input graph is dense $(d = \Theta(n))$. The classical run-time for motif clustering in this case is $O(n^3)$, but the quantum run-time is $\tilde{O}(n^{2.5})$. 

Note that, in contrast, constructing the motif graph approximately doesn't generally help us in the classical case: if we were to estimate the weights of the edges of the motif graph with constant relative error, then this would take time $\tilde{O}(nd^{l + s - 1})$, which already for $l=1$ is worse than even the exact version described above. Hence, using approximate counting only buys us something in the quantum case. 

\paragraph{Quantum via approximate counting and quantum spectral clustering} 
It is sometimes more efficient to provide query access to the approximate motif graph $\fG$ rather than to construct it explicitly beforehand. We show the following in  Section~\ref{sec:clustering_qcount_quantum_spectral_clustering}, by combining the quantum spectral clustering algorithm of Apers and de Wolf~\cite{apers20} with our algorithms for approximately constructing the motif graph via quantum counting
\begin{theorem}[Motif clustering via quantum counting \& quantum clustering]\label{theo:third_result}
Given a graph $G$ with maximum degree $d$ and a motif $M$ of size $s$, there exists a quantum algorithm for \textit{approximate} motif clustering that takes time
\[
    \tilde{O}\left( \sqrt{n^3 d^{s-2}} + \Tk\right)
\]
in expectation.
\end{theorem}
This run-time is independent of whether we have to pre-process the input graph or not. Hence, whenever $d^l = \omega(\sqrt{n})$ (which is the case for, for example, dense graphs), this algorithm is more efficient than the algorithm of Theorem~\ref{theo:second_result} above which constructs the approximate motif graph explicitly. 

It should be noted that, if we choose to use this algorithm for motif clustering, we can, generally speaking, only cluster using the unnormalized Laplacian, because we lose the ability to filter out the vertices that become disconnected in the motif graph -- see Section~\ref{sec:disconnected_vertices} for a discussion on this point.

\paragraph{Summary}
In all of our quantum algorithms our speedups come primarily from faster computation of the weights in the motif graph, and in one also from the application of quantum spectral clustering via~\cite{apers20}. In the worst case the speedup over the classical algorithm is minimal, but for many natural families of graphs the speedup can be reasonably large (i.e. quadratic). In Table~\ref{tab:result_summary} we summarize the (expected) complexities of the classical and our three quantum algorithms for performing motif clustering on a general graph using an arbitrary motif. 

\begin{table}[h]
\centering
\begin{tabular}{|l|l|}
\hline
\textbf{Algorithm} & \textbf{Expected run-time} \\ \hline
\multicolumn{1}{|l|}{Classical} & $O(nd^{s-1} + \Tk)$  \\ \hline
\multicolumn{1}{|l|}{Quantum-Grover (pre-process). Theorem~\ref{theo:first_result}} & $\tilde{O}(nd+\sqrt{nd^{s-1}\mM} + \Tk)$  \\ \hline
\multicolumn{1}{|l|}{Quantum-Grover (no pre-process) Theorem~\ref{theo:first_result}} & $\tilde{O}(\sqrt{nd^{s}\mM} + \Tk)$ \\ \hline
\multicolumn{1}{|l|}{Quantum-Approximate + classical cluster. Theorem~\ref{theo:second_result}} & $\tilde{O}(nd^{\frac{s}{2}+l-1} + \Tk)$ \\ \hline 
\multicolumn{1}{|l|}{Quantum-Approximate + quantum cluster. Theorem~\ref{theo:third_result}} & $\tilde{O}(\sqrt{n^3 d^{s-2}} + \Tk)$  \\ \hline
\end{tabular}%
\caption{Expected run-times of the algorithms for motif clustering the input graph $\fG$, where the motif is of size $s$, $n$ is the number of vertices of $G$, $d$ is the maximum degree of $G$, and $\mM$ is the number of motif instances in $G$. If we are given coherent rather than classical access to the input graph, then the run-time of Quantum-Grover is $\tilde{O}(\sqrt{nd^{s-1}\mM} + \Tk)$ if the adjacency lists are sorted, or $\tilde{O}(\sqrt{nd^{s}\mM} + \Tk)$ if not, where $\Tk$ denotes the time it takes to run $k$-means.}
\label{tab:result_summary}
\end{table}

Furthermore, we consider the complexity for power-law graphs, the latter being a model of many naturally occurring graph families, such as social-networks and internet graphs~\cite{faloutsos2011power}. In particular, we find that, if we take the motif to be a clique with two anchor nodes, the speedup becomes more significant as the size of the motif grows. The corresponding run-time complexities are given in Section~\ref{sec:comparison_run-times}.

\subsection{Anchor nodes}
Our algorithms for constructing the motif graph work for motifs with an arbitrary number of anchor nodes\footnote{The algorithms based on quantum approximate counting assume the motif has two anchor nodes, but, as we show in Appendix~\ref{sec:motifs_more_than_two}, this is sufficient to be able to construct the motif graph for motifs with an arbitrary number of anchor nodes.}. However, as in the classical case, clustering by means of the motif graph $\fG$ in order to obtain a clustering that approximately minimizes motif conductance or motif ratio cut in the original graph $G$ only works for motifs with two or three anchor nodes. 

In fact, in Appendix~\ref{sec:higher_order_motifs}, we argue that motif clustering by means of the motif graph should only be used for motifs with two anchor nodes --- or weighted combinations thereof, to be made precise in Appendix~\ref{sec:higher_order_motifs}. The reason for this is that the motif graph, being a graph itself, only captures pairwise relationships between vertices, and a motif with two anchor nodes exactly expresses a pairwise relationship between its anchor nodes. 

Motifs with more than two anchor nodes attempt to capture relationships between more than two vertices, and such relationships should be described by a hypergraph instead. This statement seems incompatible with the fact that Benson et al.~perform motif clustering using the motif graph for three-anchor node motifs. However, as we show in Appendix~\ref{sec:higher_order_motifs}, clustering using a three-anchor node motif is equivalent to clustering with a specific weighted combination of two-anchor node motifs. This equivalence breaks down for motifs with more than three anchor nodes.

\section{Classical algorithms for motif spectral clustering}
\label{sec:classical-mc}

In this section we follow~\cite{benson16} and describe classical algorithms for finding a partition that approximately minimizes the motif conductance or motif ratio cut\footnote{For motifs of two or three anchor nodes; we will find two anchor nodes to be sufficient for our purpose, see Appendix~\ref{sec:higher_order_motifs}.}. The algorithms consist of the following two steps: given the graph $G$ and a motif $M$, first construct the motif graph $\fG$; second, perform $k$-means spectral clustering on $\fG$ using either the normalized (resp. unnormalized) Laplacian in order to find a partition with a low conductance (resp. ratio cut). An overview of the run-times of the motif spectral clustering algorithms discussed in this section is given in Table~\ref{tab:classical-runtimes}, where $n = |V|$ is the number of vertices and $d$ the maximum degree of the original graph $G$, $\fm = |\fE|$ is the number of edges of the motif graph $\fG$, $s = |V_M|$ is the size of the motif $M$, $k$ is the number of clusters, and $\epsilon$ is the relative accuracy with which the eigenvectors of the Laplacian are approximated.
\begin{table}[h]
\begin{center}
\begin{tabular}{c|c|c}
Algorithm & Specifics & Run-time \\
\hline
Construct $\fG$ & general & $\mO(n^s)$ \\
& bounded degree $d$ & $\tilde{\mO}(nd^{s-1})$  \\
\hline
Obtain eigenvectors $\fL$ or $\fL_{\n}$ & exact & $\mO(n^3)$ \\
& $\epsilon$-approximate & $\tilde{\mO}\left(\fm + \frac{kn}{\epsilon^2}\right)$ \\
\hline
Perform $k$-means clustering & -- & $\Tk$ \\
\end{tabular}
\end{center}
\caption{Run-times of algorithms used for classical motif spectral clustering.}
\label{tab:classical-runtimes}
\end{table}

Note that $\fm \leq nd^l \leq nd^{s-1}$, since any two anchor nodes of a given motif instance are in each others $l$-hop neighborhood, and trivially, also $\fm \leq n^2$. If we use the $\epsilon$-approximate eigenvectors to perform spectral clustering (for which we only require constant $\epsilon >0$, see Section~\ref{sec:k-means-spectral}, and $k$ is also constant), the time to perform the $\epsilon$-approximate $k$-means step is $\tilde{O}(nd^l)$, and therefore the construction of the motif graph becomes the bottleneck in the complexity of the entire computation -- assuming $k$-means runs in nearly-linear time.

\subsection{Constructing the motif graph}
\label{sec:constructing_motif_graph}

Let $G$ be an $n$-vertex, $m$-edge graph, and $M$ be a motif of size $s$. In order to construct the motif adjacency matrix $\fA$, we need to find all instances of $M$ in $G$. The most straightforward (and in general the optimal) way to find all motif instances is to simply consider all ${n \choose s}$ $s$-sized subsets of the vertex set $V$ and check whether they form motifs for every choice of anchor nodes in each subset. Checking if a given subset of $s$ vertices forms a motif instance requires $O(s^2) = O(1)$ checks, since $s$ is constant. Therefore, the entire process of finding all motif instances takes $O(n^s)$ time.

If we have adjacency list access to $G$ and know that it has maximum degree $d$, then we can more efficiently search for possible motif instances: since the motif is always taken to be a connected graph, every motif instance can be found by (i) picking an initial vertex $u \in V$, and (ii) growing the motif instance from $u$ by exploring the local neighbourhood of $u$ in order to find $s-1$ more vertices that might yield a match to the motif. Since each vertex has degree at most $d$, there are at most $d^{s-1}$ possible choices of $s$ vertices in the local neighbourhood of each vertex, and hence this process requires $\mO(nd^{s-1})$ time to check all connected $s$-tuples of nodes for motif instances. We describe a classical procedure for constructing these subsets in Section~\ref{sec:exploring_neighbourhood}.

\ 

\noindent The complexities presented above hold for general motifs. However, for certain motifs the time to find all instances can sometimes be faster -- for example, all triangles in a graph can be found using $\Theta(m^{1.5})$ queries~\cite{latapy2008main}; all induced and non-induced `position-aware' motifs of size at most 4 can be found using $\Theta(m^2)$ queries~\cite{marcus2010efficient}; and quadrangles can be found using $\Theta(m^{1.5})$ queries~\cite{chiba1985arboricity} -- see the appendix of~\cite{benson16} for more details.

\subsection{k-means spectral clustering}
\label{sec:k-means-spectral}

Given an integer $k\in \N$, and (say, adjacency list) access to $\fA$, we can next proceed to search for a partition $W_1, \ldots, W_k$ that minimizes either the motif ratio cut or the motif conductance of $G$ by minimizing the ordinary ratio cut or conductance in $\fG$ as described at the end of Section~\ref{sec:motif-graph}. Both tasks, which are NP-hard for worst-case instances~\cite{wagner1993between}, can be tackled using k-means spectral clustering\footnote{k-means spectral clustering solves a relaxed version of the NP-hard conductance or ratio cut minimization problem. It outputs clusters that are close to the optimal clusters for \emph{well-clustered} graphs~\cite{peng2015partitioning}.} (see~\cite{luxburg2007tutorial} and references therein), which finds partitions that approximately minimize either the ratio cut or the conductance.

Whether spectral clustering minimizes ratio cut or conductance depends on whether it is performed using the ordinary or the normalized Laplacian of the motif graph. Let $\fD$ be the diagonal weighted \emph{motif degree matrix} of $\fG$, with coefficients $\fD_{uu} =\fs_u $, where $\fs_u$ is the strength of vertex $u$, and define the \emph{motif Laplacian} $\fL$ by
\[
    \fL := \fD - \fA,
\]
and the \emph{normalized motif Laplacian} by
\[
    \fL_{\n} := \fD^{-\frac{1}{2}} \fL \fD^{-\frac{1}{2}} = I - \fD^{-\frac{1}{2}} \fA \fD^{-\frac{1}{2}}.
\]

\subsubsection{Minimizing ratio cut}

In order to find a partition with small ratio cut, we can perform spectral clustering using the unnormalized Laplacian $\fL$. This works as follows.
 \begin{enumerate}
    \item Compute the $k$ eigenvectors of $\fL$ corresponding to the $k$ smallest eigenvalues, and let $U$ be the $n \times k$ matrix containing the first $k$ eigenvectors as columns. 
    \item For $i \in [n]$, let $u_i$ be the $i$-th row of $U$. Each $k$-dimensional row vector $u_i \in \R^k$ can be thought of as a feature vector for the $i$-th vertex of $V$.
    \item Cluster the vertices of $V$ by performing $k$-means clustering on the $n$ feature vectors $\{u_i\}_{i=0}^{n-1} \subset \R^k$.
\end{enumerate}

\subsubsection{Minimizing conductance}
If, instead of ratio cut, we want to find a partition with low conductance, we can apply spectral clustering to the normalized Laplacian $\fL_{\n}$:
 
\begin{enumerate}
    \item Compute the first $k$ eigenvectors of $\fL_{\n}$ corresponding to the smallest $k$ eigenvalues, and let $\tilde{U}$ be the $n \times k$ matrix containing the first $k$ eigenvectors as columns. 
    \item Let $U$ be the matrix obtained by taking $\tilde{U}$, and renormalizing all the rows to 1, that is: 
    \[
        u_{ij} := \frac{\tilde{u}_{ij}}{\sqrt{\sum_{l=1}^{k} \tilde{u}_{il}^2}}.
    \]
    \item For $i \in [n]$, let $u_i$ be the $i$-th row of $U$. Each $k$-dimensional row vector $u_i \in \R^k$ can be thought of a feature vector for the $i$-th vertex of $V$.
    \item Cluster the vertices of $V$ by performing $k$-means clustering on the $n$ feature vectors $\{u_i\}_{i=0}^{n-1} \subset \R^k$.
\end{enumerate}
Note that, for $\fd$-regular graphs, $\vol_{\fG}(W) = \fd|W|$. As a consequence, for these graphs minimizing the conductance is equivalent to minimizing ratio cut.

\subsubsection{Disconnected vertices}
\label{sec:disconnected_vertices}

It is possible for certain vertices in $G$ to become entirely disconnected in $\fG$ because their motif degree is zero. If we want to use the normalized Laplacian for spectral clustering, then we first have to remove all such vertices since the normalized Laplacian is obtained by multiplying the original Laplacian by $D^{-\frac{1}{2}}$. 

Moreover, if we were to perform spectral $k$-means clustering (using the unnormalized Laplacian) on $\fG$ with $\fV = V$, the algorithm could just output several clusters containing a single disconnected vertex each and place the remaining vertices into one or more larger clusters to minimize ratio cut. From the perspective of the motif adjacency matrix, the clusters containing a single disconnected vertex are not very interesting. Hence, after the construction of $\fG$ we should remove all vertices $u \in V$ that have zero (motif) degree $\fd_u = 0$ and put each of them in their own size-one cluster. The remaining vertex set $\{u \in V: \fd_u > 0\}$ will then be the vertex set of $\fG$ on which we perform k-means spectral clustering. Note that this procedure yields a number of clusters that is equal to $k$ plus the number of vertices in $V$ that are no anchor node of any motif instance in $G$. 

In the remainder of this work, we will not emphasise this practical detail and simply write $V$ for the vertex set of $\fG$. Note that removing disconnected vertices does not affect our run-time upper bounds, since none of our algorithms run in sub-linear time (a single step of $k$-means takes time linear in $n$).

\subsubsection{Complexity}

For an $n$-vertex, $\fm$-edge motif graph $\fG$, it is possible to compute the eigenvectors of $\fL$ or $\fL_{\text{norm}}$ in $O(n^3)$ time via exact diagonalization. However, we can also use $\epsilon$-approximate spectral clustering to find an an $\epsilon$-approximation to the $k$ smallest eigenvalues $\lambda_1, \ldots, \lambda_k$ of the (normalized) Laplacian together with a set of orthonormal unit vectors $v_1, \ldots, v_k$ such that
\[
    v_i^T \fL_{(\text{norm})} v_i \leq (1 + \epsilon) \lambda_i \quad \text{for all} \quad 1 \leq i \leq k
\]
in time at most $\tilde{O}\left(\fm + \frac{k\fn}{\epsilon^2} \right)$~\cite{spielman2007spectral, spielman2014nearly, koutis2015faster}. This set of unit vectors approximates the subspace spanned by the $k$ smallest eigenvectors of $\fL$, and is suitable for performing spectral clustering, even for \emph{constant} $\epsilon > 0$~\cite{peng2015partitioning, apers20}. Finally, we use $k$-means, which takes time $\Tk$ adding up to a total run-time of $\tilde{O}(\fm+\Tk)$, since one step of k-means already takes time $nk^2$ and $\epsilon$ is constant.

As described in the references above, the method for finding approximate eigenvectors makes use of a graph sparsification algorithm which, given the graph $\fG$, constructs a spectral sparsifier $\fG_S$ of $\fG$, and then uses the inverse power method on the graph Laplacian $\fL_{S}$ corresponding to $\fG_S$. This method can also be used to construct approximate eigenvectors of the normalized Laplacian $\fL_{\n}$, by applying the inverse power method to $\fD^{-1/2} \fL_{S} \fD^{-1/2}$, where $\fD$ is the degree matrix of the unsparsified graph $\fG$~\cite{koutis2015faster}.

\section{Quantum motif clustering}\label{sec:quantum_algorithms}
\label{sec:quantum_motif_clustering}

In this section we present three quantum algorithms for motif spectral clustering, one using Grover search, and the other two using quantum approximate counting. All three algorithms consist of two steps: (i) construct the motif graph $\fG$, and (ii) perform spectral clustering on $\fG$. As discussed in Section~\ref{sec:results}, the bottleneck for motif clustering is in step (i), constructing the motif graph, and this is also where our contribution lies; for step (ii) we use either the spectral clustering algorithm based on the spectral estimation algorithm of~\cite{spielman2014nearly} or the algorithm for quantum spectral clustering of Apers and de Wolf~\cite{apers20}.

We begin by introducing the quantum tools that we make use of, followed by a description in Section~\ref{sec:quantum_prelims} of a classical subroutine for exploring the local neighbourhood of vertices in a graph according to a particular motif structure. In Section~\ref{sec:motif_grover} we discuss how Grover search can be applied to find all motif instances in order to do motif clustering, and in Section~\ref{sec:motif_clustering_quantum_counting} we use quantum approximate counting to construct an approximation to the motif graph, and discuss under what conditions this approximation can be used for motif clustering. We then compare the run-times of all approaches in Section~\ref{sec:comparison_run-times}. Subsequently, in Section~\ref{sec:approx-adjacency-matrix}, we provide details to justify the use of approximations in the context of spectral clustering.

\subsection{Preliminaries}
\label{sec:quantum_prelims}

We will find the following quantum subroutines useful. For each, we consider a Boolean function $f : [N] \rightarrow \{0,1\}$ on $N$ items, with $t = |\{i : f(x) = 1\}|$ the (unknown) number of `marked' items. We will assume that we have oracle access to $f$, i.e. a unitary $\mathcal{O}_f$ that acts as $\mathcal{O}_f \ket{x}\ket{0} = \ket{x}\ket{f(x)}$.  

\begin{lemma}[Grover search with an unknown number of marked items~\cite{boyer1998tight}]
\label{lem:grover}
There exists a quantum algorithm $\text{\bf Search}(\mathcal{O}_f)$ that, with probability at least $2/3$, finds and returns an index $i \in \{0,\dots,N-1\}$ such that $x_i = 1$ if one exists, and requires an expected number $O(\sqrt{N/t})$ queries to $\mathcal{O}_f$ and its inverse, and $O(\sqrt{N/t}\log(N))$ other elementary operations. If no such $x_i$ exists, the algorithm indicates this with certainty and requires $O(\sqrt{N})$ queries to $\mathcal{O}_f$ and its inverse, and $O(\sqrt{N}\log(N))$ other elementary operations.
\end{lemma}

Using $\text{\bf Search}(\mathcal{O}_f)$, one can output \emph{all} $t$ marked items in time $\tilde{O}(\sqrt{Nt})$. More precisely, we have 
\begin{lemma}[Finding all marked items~\cite{boyer1998tight}]
\label{lem:find_all}
There exists a quantum algorithm $\text{\bf Find}(\mathcal{O}_f)$ that outputs, with constant probability, all $t$ marked items and makes an expected $O(\sqrt{N t})$ calls to $\mathcal{O}_f$ and its inverse, and $\tilde{O}(\sqrt{N t})$ other elementary operations. 
\end{lemma}

\begin{lemma}[Approximate Quantum Counting~\cite{brassard2002quantum}]
\label{lem:approx_counting}
    There exists a quantum algorithm $\text{\bf Approx-Count}(\mathcal{O}_f,\epsilon,\delta)$ that, with probability $\geq 1-\delta$, outputs a number $\tilde{t}$ such that
    \[
        |\tilde{t} - t| \leq \epsilon t
    \]
    using an expected number $O\left(\frac{1}{\epsilon}\sqrt{N/t} \log(1/\delta) \right)$ of calls to $\mathcal{O}_f$ and $\mathcal{O}_f^{-1}$ and $\tilde{O}\left(\frac{1}{\epsilon}\sqrt{N/t} \log(1/\delta) \right)$ other elementary operations. If $t=0$, then the algorithm outputs with certainty $\tilde{t}=t$ and calls $\mathcal{O}_f$ and its inverse $O(\sqrt{N})$ times, and uses $\tilde{O}(\sqrt{N})$ other elementary operations. 
\end{lemma}


Finally, we will use a quantum algorithm for spectral $k$-means clustering from Apers and de Wolf~\cite{apers20}, which itself uses the (classical) algorithm of Spielman and Teng~\cite{spielman2014nearly} to find approximations to the first $k$ eigenvectors of a graph Laplacian obtained via quantum graph sparsification.
\begin{lemma}[Quantum spectral estimation~\cite{apers20}]
    Given adjacency list access to an $n$-vertex weighted graph $G$ with $m$ edges, there exists an $\tilde{O}(\sqrt{mn}/\epsilon + kn/\epsilon^2)$-time quantum algorithm that outputs, with high probability, an $\epsilon$-approximation of each of the $k$ smallest eigenvalues $\lambda_1, \ldots, \lambda_k$ of the graph Laplacian $L$, and a set of orthogonal unit vectors $v_1,\dots,v_k$ such that $v_l^T L v_l \leq (1+\epsilon)\lambda_l$ for all $1 \leq l \leq k$. 
\end{lemma}
\noindent It turns out that choosing $\epsilon$ to be constant is already enough to perform spectral clustering~\cite{apers20}, and hence Apers and de Wolf note that 
\begin{corollary}[Quantum spectral clustering~\cite{apers20}]\label{cor:qsc_apers_dewolf}
    There exists a quantum algorithm that, given adjacency list access to an $n$-vertex weighted graph with $m$ edges, performs spectral $k$-means clustering on the graph in time $\tilde{O}(\sqrt{mn} + \Tk)$. 
\end{corollary}
\noindent Classically, a fast algorithm for (approximate) $k$-means spectral clustering can be obtained by combining the spectral estimation routine of Spielman and Teng~\cite{spielman2014nearly} with constant error with a $k$-means clustering algorithm, yielding the following:
\begin{lemma}[Spectral clustering~\cite{spielman2014nearly}]\label{thm:sc_spielman}
    There exists a classical algorithm that, given adjacency list access to an $n$-vertex weighted graph with $m$ edges, performs spectral $k$-means clustering on the graph in time $\tilde{O}(m + \Tk)$. 
\end{lemma}

\subsubsection{Exploring the `motif neighbourhood' of a vertex}\label{sec:exploring_neighbourhood}
Here we describe a short classical algorithm which, given a vertex $u$, motif $M$ of size $s$, and a sequence of integers $I$ of length $s-1$, can be used to return a pairing of $s$ vertices around $u$ to vertices in $M$, such that those vertices are candidates for a match of the motif in the graph. We call this procedure a `tree walk', for reasons that will become apparent. More precisely, given a tree $T$ of $t$ vertices, a graph $G$, and a vertex $v$, a tree walk explores the neighbourhood of $v$ in $G$ by constructing the tree $T$ locally out of the neighbours of $v$. The output is a list of size $t$ that identifies vertices in $G$ with vertices in the tree $T$. We give details of the tree walk in Algorithm~\ref{algo:treewalk}, and show an example of two outcomes of a tree walk in Figure~\ref{fig:tree_walks}. For any fixed input, the tree walk algorithm takes time linear in the number of edges in the tree $T$. Note that it is possible for Algorithm~\ref{algo:treewalk} to return a list $L$ shorter than desired ($T$). This will not be a problem for us.

\begin{algorithm}[H]
    \caption{Tree walk}
    \begin{algorithmic}[1]
    \Function{Tree walk}{Tree $T$ with root $r$, graph $G$, vertex $v$, sequence $I = \{i_j\}_{j=1}^{|T|}$}
        \State Initialise $L \gets \{(r,v)\}$
        \For{Each child vertex $c$ of $r$}
            \State Let $i_j$ be the first element from $I$
            \State Let $u$ be the $i_j$th neighbour of $v$ in $G$ (if $v$ has no $i_j$th neighbour, set $u = \O$)
            \State Let $T_c$ be the sub-tree rooted at $c$
            \State $I \gets I \setminus \{i_j\}$
            \State $L' \gets$ \Call{Tree walk}{$T_c$, $G$, $u$, $I$}
            \State $L \gets L \cup L'$
        \EndFor
        \State \Return $L$
    \EndFunction

    \end{algorithmic}
    \label{algo:treewalk}
\end{algorithm}

\begin{figure}
    \centering
    \includegraphics[scale=0.8]{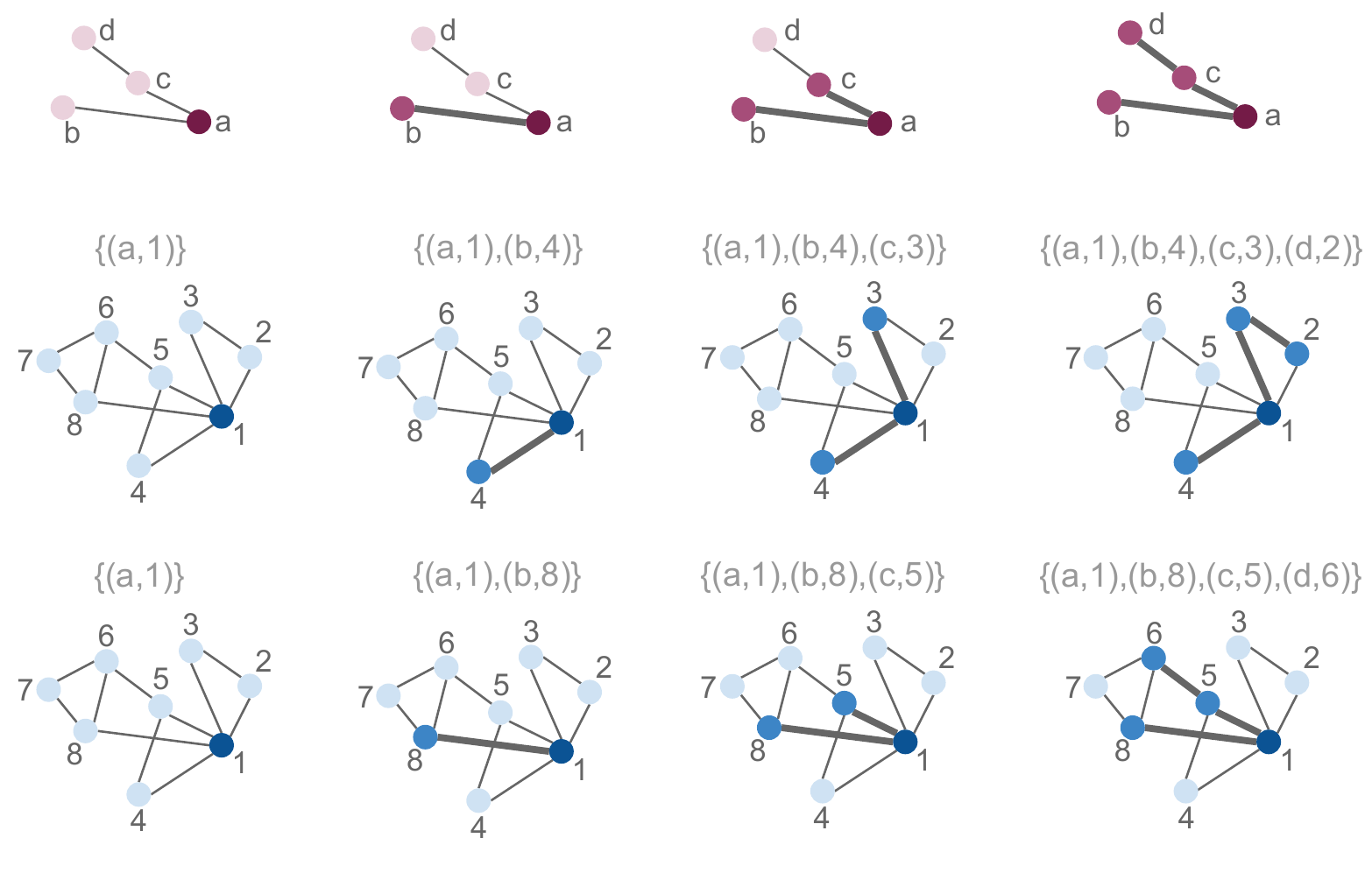}
    \caption{The stages of two possible tree walks (corresponding to different settings of the integer sequence), showing the list $L$ maintained during the tree walk after each step. The tree is shown in red (top), the graph on which the walk is performed in blue (middle and bottom). The edges and vertices selected at each step are shown (thicker edges and darker vertices, respectively).}
    \label{fig:tree_walks}
\end{figure}

\subsection{Motif clustering with Grover search}
\label{sec:motif_grover}

The most straightforward way to speed up classical motif clustering is to replace the search for motif instances by a Grover search. In doing so, we find all motif instances in $G$ and construct $\fG$ exactly, and provide a generic speedup over the classical approach. The advantage of producing the adjacency lists exactly is that we can then cluster based on both the normalized and unnormalized motif Laplacian with the usual theoretical guarantees. In Sections~\ref{sec:6.2.1} and~\ref{sec:6.2.2} below we prove Theorem~\ref{theo:first_result}.

\paragraph{Checking for motif matches in the graph} 
We will routinely need to check if a set of $s$ vertices in the input graph $G$ corresponds to a match of the motif $M$. To do this, we need to check if all the edges (resp. non-edges) of $M$ are present (resp. not present) in $G$. Since we only assume adjacency list access to $G$, this will incur some overhead. In particular, to check for the presence of an edge $(u,v)$ in $G$, we can check if $v$ appears as a neighbour in $u$'s adjacency list. 

As discussed in the introduction, if the adjacency lists are sorted (according to some fixed ordering over the vertices of $G$), then this can be done in time $O(\log d)$ per edge. since we will have to check the existence or non-existence of $O(s^2)$ edges -- with $s$ constant -- this will take time $O(s^2 \log d) = O(\log d)$.

If the adjacency lists are not sorted, then we can instead use a single application of \textbf{Search} from Lemma~\ref{lem:grover} to detect the presence of the edge with probability $\geq 1-\epsilon$ in time $O(\sqrt{d}\log(1/\epsilon))$. If we apply this subroutine $N\geq 1$ times and we want all the instances as a whole to succeed with probability at least $1-C$ ($C>0$ constant to choose to your liking), then we need $\epsilon = C/N$ by the union bound. The number of times the function \Call{Match}{} is called is $N=O(\sqrt{nd^{s-1}\fM})$ in Algorithm \ref{algo:motif_grover}, which is also the number of times the subroutine for checking edges is called (since $s$ is constant). Because this is at most polynomial in $n$, it will only add a logarithmic overhead to the run-time of \textbf{Search} for edges, which therefore takes time $\tilde{O}(s^2 \sqrt{d}) = \tilde{O}(\sqrt{d})$.

\subsubsection{Constructing $\fG$ via Grover search}\label{sec:6.2.1}
Our first algorithm is a basic application of Grover search to find all matches of the motif within the graph, which with some short (classical) post-processing allows us to obtain the motif adjacency lists $\{\fJ_u\}_{u \in V}$ exactly (i.e. without errors on the weights). As we construct the motif adjacency lists, we keep them ordered according to some arbitrary but fixed ordering of the vertices in $V$. The algorithm makes direct use of the TreeWalk sub-routine from Algorithm~\ref{algo:treewalk}.
\begin{algorithm}[H]
    \caption{Motif-Grover}
    \begin{algorithmic}[1]
    \State Initialise $\fJ_u$ for each $u \in V$.
    \State Construct a spanning tree $T$ for the motif $M$.
    \State Let $\mathcal{O}_f$ be the unitary that implements \Call{Match}{}, taking as input some vertex, integer-sequence pair $(u,I)$ (where $I = \{i_j\}_{j=1}^{s-1}$, $i_j \in \{0,\dots,d-1\}$). 
    \State Use $\text{\bf Find}(\mathcal{O}_f)$ to find all $\fM$ motif instances in the graph, searching over all $O(n d^{s-1})$ possible vertex/integer-sequence pairs $(u,I)$. 
    \State Classically post-process the $\fM$ matches to fill in the entries of each $\fJ_u$. 

    \State \Return $\{\fJ_u\}_{u \in V}$
    
    \item[]
    
    \Function{Match}{Tree $T$, graph $G$, vertex $v$, sequence $I$}
        \State $L \gets $ \Call{TreeWalk}{$T$,$G$,$v$,$I$}\label{small_hippo_grover}
        \If{$L$ gives a match to the motif in $G$}\label{giant_camel_grover}
            \State \Return 1
        \Else 
            \State \Return 0
        \EndIf
    \EndFunction
    \end{algorithmic}
    \label{algo:motif_grover}
\end{algorithm}
Each tree walk (Line~\ref{small_hippo_grover}) takes time $s-1$ and each check for a match (Line~\ref{giant_camel_grover}) takes time either $O(\log d)$ or $\tilde{O}(\sqrt{d})$, depending on whether the adjacency lists are sorted or unsorted. From Lemma~\ref{lem:find_all}, $\text{\bf Find}$ requires an expected $O(\sqrt{nd^{s-1}\fM})$ applications of these subroutines and $\tilde{O}(\sqrt{nd^{s-1}\fM})$ other operations. The classical post-processing takes $O(l \log d)$ time for each match (if the match concerns vertices $u$ and $v$ as anchor nodes, we may have to do binary search of the motif neighbours of $u$ to find the entry in $u$'s motif adjacency list that corresponds to vertex $v$, and vice versa), hence $\tilde{O}(\fM)$ time in total. All steps combined, the algorithm takes time $\tilde{O}(\sqrt{nd^{s-1}\fM} + \fM) = \tilde{O}(\sqrt{nd^{s-1}\fM})$ if the adjacency lists are sorted, and $\tilde{O}(\sqrt{nd^{s}\fM})$ time if they are not. If we choose to pre-sort the adjacency lists ahead of time, or need to load the input graph to QRAM, this will add an extra additive overhead of $\tilde{O}(nd)$. 

If the motif $M$ is symmetric under non-trivial motif isomorphisms, then Algorithm~\ref{algo:motif_grover} will find each motif instance exactly $S_M$ times, where $S_M$ is the number of motif isomorphisms of $M$. Because $s$ is constant, so is $S_M$, and therefore we incur a constant overhead in the presence of motif symmetries; see Appendix~\ref{sec:motif_isomorphisms} for details.

\subsubsection{Clustering using Motif-Grover}\label{sec:6.2.2}
In the previous section we established how to obtain adjacency list access to the (exact) motif graph $\fG$. To perform motif clustering, we apply $k$-means spectral clustering to $\fG$ using the spectral clustering algorithm of Lemma~\ref{thm:sc_spielman}, which results in a clustering that approximately minimizes the motif RatioCut (when applied to the ordinary Laplacian of $\fG$), or the motif conductance (when applied to the normalized Laplacian of $\fG$).

\begin{algorithm}[H]
	
	\caption{Quantum motif clustering: Grover version}
	
	 \hspace*{\algorithmicindent} \textbf{Input}: A graph $G=(V,E)$, motif $M$, and an integer $k\geq 1$. 

	\begin{algorithmic}[1]

		\State Apply Algorithm \ref{algo:motif_grover} on input $G=(V,E)$ and motif $M$, which returns the set of exact adjacency lists $\{\fJ_u\}_{u \in V}$ describing the motif graph $\fG$.
		\State Apply the algorithm from Lemma~\ref{thm:sc_spielman} to either i) the unnormalized or ii) the normalized motif graph Laplacian, using the adjacency lists to provide access.
		\State \Return A $k$-partition of $G=(V,E)$ that approximately minimizes either i) motif RatioCut or ii) Motif conductance.
	\end{algorithmic}
	\label{algo:motif_cluster_grover}
\end{algorithm}

The run-time of step 2 is $\tilde{O}(\fm + \Tk)$, where $\fm \leq \text{min}(nd^l, n^2)$ is the number of edges in the motif graph $\fG$ (recall that $l$ is the maximum distance between any two anchor nodes in the motif). Since $\fm \leq \fM = O(nd^{s-1})$, with similar considerations to the algorithm from the previous section, the total expected run-time is either $\tilde{O}(\sqrt{nd^{s-1}\fM} + \fm + \Tk) = \tilde{O}(\sqrt{nd^{s-1}\fM} + \Tk)$ if the adjacency lists are sorted, else $\tilde{O}(\sqrt{nd^{s}\fM} + \Tk)$. Once again, if we either have classical access to the input graph and need to load it to QRAM first, or we need pre-sort the adjacency lists ahead of time in order to apply the second algorithm, we add an extra additive overhead of $\tilde{O}(nd)$.

\subsection{Motif clustering via quantum counting}
\label{sec:motif_clustering_quantum_counting}

If some reasonably weak conditions hold for the motif, then it can be faster to use quantum counting to obtain an $\epsilon$-approximation of the motif graph (in the sense that the weights on the edges are approximated up to relative error $\epsilon$), and then use this for motif clustering. In many cases a rough approximation to the graph is good enough for clustering, and we argue this both formally (in the case of clustering on the unnormalized Laplacian) and empirically (in the case of clustering on the normalized Laplacian) in Section~\ref{sec:approx-adjacency-matrix}. 

In this section we present a quantum algorithm (Algorithm~\ref{algo:motif_approxcount}) for constructing the approximate motif graph using quantum approximate counting. As before, this algorithm can be combined with the algorithm of either Corollary~\ref{cor:qsc_apers_dewolf} or Lemma~\ref{thm:sc_spielman} to obtain a quantum algorithm for (approximate) motif clustering. We begin by describing a quantum algorithm (Algorithm~\ref{algo:motif_count}) for computing approximations to the entries of the motif adjacency matrix $\fA$, and then use this to construct the motif graph. More precisely, given vertices $u$ and $v$, we provide a quantum algorithm that outputs an approximation $\widetilde{\fA_{uv}}$ satisfying
\begin{equation}\label{eq:approx_A}
    \left| \widetilde{\fA_{uv}} - \fA_{uv} \right| \leq \epsilon \fA_{uv}
\end{equation}
with probability at least $1-\delta$ for some choice of accuracy $\epsilon$ and probability of failure $\delta>0$.

The algorithm for approximate motif clustering described in this section assumes the motif $M$ has two anchor nodes. However, our algorithm can easily be extended to motifs with more than 2 anchor nodes, since the motif graph can be constructed in this case by decomposing the motif into a combination of two-anchor-node motifs -- see Appendix~\ref{sec:motifs_more_than_two} for details.

\begin{algorithm}[H]
    \caption{Motif Count}
    \begin{algorithmic}[1]
    \Function{Motif Count}{Motif $M=(V_M,E_M,V_A)$ with two anchor nodes $a$ and $b$, graph $G$, vertices $u, v$, accuracy $\epsilon$, probability $1-\delta$}
    
        \State Construct a spanning tree $T=(V_T,E_T)$ of $M$.
        \State Split $T$ into two trees: an $s_a$-vertex tree $T_a=(V_{T_a},E_{T_a})$ rooted at $a \in V_A$ and an $s_b$-vertex tree $T_b=(V_{T_b},E_{T_b})$ rooted at $b \in V_A$.
        
        \State Let $O_f$ be the unitary that implements \Call{Match}{} on the trees $T_a, T_b$, graph $G$, vertices $u, v$, and taking as input two integer sequences $I_a = \{i^{(a)}_j\}_{j=1}^{s_a-1}, I_b = \{i^{(b)}_j\}_{j=1}^{s_b-1}$, where $s_a = |V_{T_a}|$ and $s_b = |V_{T_b}|$.
        \State \Return $\text{\bf{Approx-Count}}(O_f,\epsilon,\delta)$
    \EndFunction
    
    \item[]
    
    \Function{Match}{Trees $T_a, T_b$, graph $G$, vertices $u,v$, sequences $I_a,I_b$}
        \State $L_a \gets $ \Call{TreeWalk}{$T_a$,$G$,$u$,$I_a$}\label{small_hippo}
        \State $L_b \gets $ \Call{TreeWalk}{$T_b$,$G$,$v$,$I_b$}\label{medium_elephant}
        \State $L \gets L_a \cup L_b$
        \If{$L$ gives a match to the motif in $G$}\label{giant_camel}
            \State \Return 1
        \Else 
            \State \Return 0
        \EndIf
    \EndFunction
    \end{algorithmic}
    \label{algo:motif_count}
\end{algorithm}

\begin{lemma}\label{lem:approx_adjacencymatrix}
Suppose we have coherent adjacency list access to an $n$-vertex graph $G$ with maximum degree $d$, where the adjacency lists are assumed to be sorted, and we are given an $s$-vertex motif $M = (V_M,E_M,V_A)$ with $|V_A| = 2$. Let $u, v$ be two vertices of $G$, and $\epsilon>0$. Then, if $\fA_{uv} \geq 1$, Algorithm~\ref{algo:motif_count} outputs an approximation $\widetilde{\fA_{uv}}$ satisfying Eq.~\eqref{eq:approx_A} with probability $\geq 1- \delta$ using $O\left(\frac{1}{\epsilon}\sqrt{\frac{d^{s-2}}{ \fA_{uv}}} \log(1/\delta)  \log(d)\right)$ queries to the graph $G$, and $\tilde{O}\left(\sqrt{\frac{d^{s-2}}{\fA_{uv}}}\log(1/\delta)  \log(d)\right)$ other elementary operations. If $\fA_{uv} = 0$, then Algorithm~\ref{algo:motif_count} outputs $\widetilde{A}_{uv} = 0$ with certainty using $O(\sqrt{d^{s-2}})$ queries to the graph $G$ and $\tilde{O}(\sqrt{d^{s-2}})$ other operations.
\end{lemma}
\begin{proof}
    Our task is to approximately count the number of motifs present in the graph $G$ that both $u$ and $v$ appear in (as anchor nodes). Algorithm~\ref{algo:motif_count} achieves this by using tree walks to explore locally the areas around $u$ and $v$, in search of a set of vertices and edges that match the motif structure, and then uses approximate quantum counting to estimate the number of motifs containing $u$ and $v$. 
    
    As described in Algorithm~\ref{algo:motif_count}, we start by (classically) constructing a spanning tree $T = (V_T, E_T)$ of the motif $M$, for example using breadth-first search in time $O(s^2) = O(1)$. Next, we remove an edge of the tree in such a way that the two newly formed trees contain one anchor node $a$ and $b$ each, which yields an $s_a$-vertex tree $T_a = (V_{T_a}, E_{T_a})$ rooted at $a$ and an $s_b$-vertex tree $T_b = (V_{T_b}, E_{T_b})$ rooted at $b$, respectively. Note that $|E_{T_a}| = s_a-1$, $|E_{T_b}| = s_b-1$, and $s_a + s_b = s$.  See Figure~\ref{fig:motif_tree} for an example.
    
    \begin{figure}
        \centering
        \includegraphics{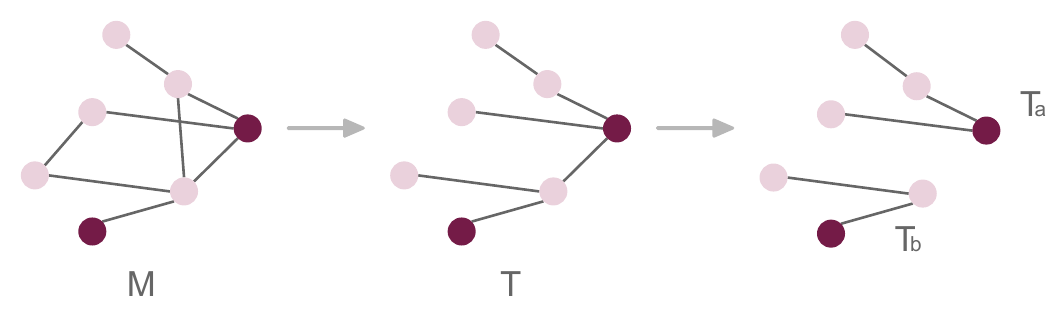}
        \caption{A motif $M$ with two anchor nodes (dark red), the spanning tree $T$ of $M$, and the two trees ($T_a$ and $T_b$) rooted at the anchors nodes.}
        \label{fig:motif_tree}
    \end{figure}
    
    We then fix two integer sequences that uniquely define a tree walk on each tree: $I_a = \{i^{(a)}_j\}_{j=1}^{s_a-1}$ for the tree walk on $T_a$ and $I_b = \{i^{(b)}_j\}_{j=1}^{s_b-1}$ for the tree walk on $T_b$. For a fixed pair of sequences, we require $s_a + s_b - 2 = s-2$ queries to $G$ to perform both tree walks. Let $L = \{(b_j,v_j)\}_{j=1}^s$ be the union of the lists output by the two tree walks (i.e. where each $b_j$ labels a vertex in $M$ and $v_j$ a vertex in $G$). To check whether the subset of vertices $v_j$ match the motif, we need to check that for every edge $(b_j,b_{j'}) \in E_M$, $(v_j,v_{j'}) \in E$ and for every non-edge $(b_j,b_{j'}) \notin E_m$, $(v_j,v_{j'}) \notin E$. Since we only assume adjacency list access to the edges of $G$, this incurs some overhead. In particular, given vertices $v_0,v_1$ from $G$, we must query all of the neighbours of $v_0$ for the presence (resp. non-presence) of $v_1$ to check for the edge (resp. non-edge) $(v_0,v_1)$. Since the adjacency lists are sorted, this can be done in time $O(\log d)$ (recall that $s$ is constant).

    For the two trees $T_a$ and $T_b$, there are at most $d^{s_a-1} d^{s_a-1} = d^{s_a + s_b - 2} = d^{s-2}$ possible tree walks that can be performed on $G$, each corresponding to a unique string of $s-2$ integers of value at most $d$. Hence \textbf{Approx-Count} will search over at most $d^{s-2}$ items. With accuracy $\epsilon$ and probability of success $\delta$, if $\fA_{uv} \geq 1$, this requires $O\left(\frac{1}{\epsilon} \sqrt{\frac{d^{s-2}}{\fA_{uv}}} \log(1/\delta) \right)$ queries to the subroutine that checks if a subset of vertices in the graph matches the motif structure, plus $\tilde{O}\left(\frac{1}{\epsilon} \sqrt{\frac{d^{s-2}}{\fA_{uv}}} \log(1/\delta)\right)$ other operations; if $\fA_{uv} = 0$, we need $O(\sqrt{d^{s-2}})$ queries and $\tilde{O}(\sqrt{d^{s-2}})$ other operations.
    
    Finally, in the presence of symmetries within the motif, we note that the quantum counting routine will over count-motif matches. In Appendix~\ref{sec:motif_isomorphisms} we work out exactly how many duplicates will be found, and show that this quantity, $S_M^{(a,b)}$, which is $O(1)$ because $s = O(1)$, depends only on the motif $M$ itself and can be computed ahead of time. Dividing the output of \textbf{Approx-Count} by $S_M^{(a,b)}$ we obtain our estimate $\widetilde{\fA_{uv}}$ (note that this doesn't affect the accuracy of the estimate).
\end{proof}

\subsubsection{Constructing $\fG$ via quantum approximate counting}
As we discuss in Section~\ref{sec:approx-adjacency-matrix}, for the purpose of clustering it turns out that approximating the edge weights up to \emph{constant} relative error is sufficient, and, as we will see, provides a speedup over the classical algorithm when the motif length $l$ satisfies $2l<s$. Again we construct $\fG$ by explicitly constructing the motif adjacency lists for each vertex, as described in Algorithm~\ref{algo:motif_approxcount}.

\begin{algorithm}[H]

	\caption{Approximate-Motif}
		 \hspace*{\algorithmicindent} \textbf{Input}: A graph $G=(V,E)$, motif $M$ with two anchor nodes, relative error $\epsilon>0$ and probability of failure $\delta = \Omega\left(\frac{1}{\poly(n)}\right)$. 
	\begin{algorithmic}[1]
		\State Initialise the adjacency lists $\{\widetilde{\fJ}_u\}_{u \in V}$.
		\For{Every vertex $u \in V$}
		    \State Find the $l$-hop neighborhood $N^l_u$ of $u$ via breadth-first-search.
			\For{Every vertex $v\in N^l_u$}
				\State Use Algorithm~\ref{algo:motif_count} to obtain $\widetilde{\fA}_{uv}$, an approximation of $\fA_{uv}$ up to relative error $\epsilon$ and with probability $\geq 1-\frac{\delta}{n^2}$.
				\If{$\widetilde{\fA}_{uv} \neq 0$}
				    \State Update $\widetilde{\fJ}_u$ to include the edge to $v$ with weight $\widetilde{\fA}_{uv}$.
				\EndIf
			\EndFor
		\EndFor
		\State \Return $\{\widetilde{\fJ}_u\}_{u \in V}$.
	\end{algorithmic}
	\label{algo:motif_approxcount}
\end{algorithm}
The output of this algorithm is a classical description of an approximation of the motif adjacency lists. These lists store approximations to the non-zero entries of the motif adjacency matrix $\fA$, and they satisfy, for all $u,v\in V$, 
\begin{equation*}
	(1-\epsilon)\fA_{uv}\leq \widetilde{\fA}_{uv}\leq (1+\epsilon)\fA_{uv}.
\end{equation*}

Complexity-wise, the computation of each $\tilde{\fA}_{uv}$ (using Algorithm~\ref{algo:motif_count} as a subroutine) takes time $\tilde{O}(\frac{1}{\epsilon}\sqrt{d^{s-2}})$, while the construction and looping over of the $l$-hop neighbourhoods takes time $O(d^l)$. So for constant $\epsilon>0$ and $\delta = \Omega\left(\frac{1}{\poly(n)}\right)$, the total run-time of this algorithm is 
\begin{equation*}
    \tilde{O}(nd^l\sqrt{d^{s-2}}) \,.
\end{equation*}

\subsubsection{Motif clustering with quantum counting and classical clustering}

Given the approximate motif graph $\widetilde{\fG}$ constructed using Algorithm~\ref{algo:motif_approxcount}, we next proceed to cluster the vertex set of $\widetilde{\fG}$, with the approximate adjacency lists $\{\widetilde{\fJ}_u\}_{u \in V}$ used to provide access. 


We first consider spectral clustering on the approximate unnormalized Laplacian, meaning that the clusters aim to minimize RatioCut. The input is an unweighted graph $G=(V,E)$ on $n$ vertices and a motif $M$ of size $s$ with two anchor nodes at distance $l$ from each other. We apply Algorithm \ref{algo:motif_approxcount} to obtain the adjacency lists of the approximate motif graph $\tilde{\fG}$ up to fixed relative error $\epsilon > 0$, and then use the algorithm of Lemma~\ref{thm:sc_spielman} to cluster $\tilde{\fG}$. 

Let $\widetilde{\fD}$ be the diagonal matrix of (approximate) motif degrees obtained from $\widetilde{\fA}$, and $\widetilde{\fL} = \widetilde{\fD} - \widetilde{\fA}$ be the approximate motif Laplacian.
By Lemma~\ref{lem:psd-relation_laplacians} (which we will prove later), the spectral structure of the true motif graph Laplacian $\fL$ of $\fG$ is preserved by our approximation, i.e.
\begin{equation}\label{eq:lapl-psd-approx}
    (1-\epsilon)\fL\preceq \tilde{\fL}\preceq (1+\epsilon)\fL\,.
\end{equation}
This property is necessary for applying the Spectral Clustering algorithm in Lemma~\ref{thm:sc_spielman}, where it suffices to choose $\epsilon$ to be some small constant. Our quantum motif clustering algorithm is given in Algorithm~\ref{algo:motif_cluster_qcount} below, which can also be used to cluster using the normalized approximate motif graph Laplacian, though here we don't have the same theoretical guarantees -- see next paragraph.

\begin{algorithm}[H]
	
	\caption{Quantum motif clustering via quantum approximate counting and classical spectral clustering}
	
	 \hspace*{\algorithmicindent} \textbf{Input}: A graph $G=(V,E)$, motif $M$ with two anchor nodes and integer $k\geq 1$. 

	\begin{algorithmic}[1]

		\State Apply Algorithm \ref{algo:motif_approxcount} on input $G=(V,E)$, motif $M$ and  some fixed  relative error $\epsilon>0$. This returns the approximate motif graph $\tilde{\fG}$ via the set of approximate adjacency lists $\{\widetilde{\fJ}_u\}_{u \in V}$ up to relative error $\epsilon$.
		\State Apply the algorithm from Lemma~\ref{thm:sc_spielman} on either i) the unnormalized approximate motif Laplacian $\tilde{\fL}$ or ii) the normalized approximate motif Laplacian $\widetilde{\fL}_{\n}$ of $\widetilde{\fG}$.
		\State \Return A $k$-partition of $G=(V,E)$ that approximately minimizes either the i) motif Ratiocut or ii) motif conductance. 
	\end{algorithmic}
	\label{algo:motif_cluster_qcount}
\end{algorithm}

\noindent With $s$ the size of $M$ and $l$ the distance between the two anchor nodes, the first step takes $\widetilde{O}(nd^{l+\frac{s}{2}-1})$ time.  The run-time of the second step is at most $\widetilde{O}(\fm + \Tk)$, where $\fm$ is the number of edges in the motif graph given by $\widetilde{\fA}$. Since $\fm = O(\min(nd^l,n^2))$, the expected run-time of the entire algorithm is 
\[
\widetilde{O}(nd^{l+\frac{s}{2}-1} + \Tk)\,,
\]
hence proving Theorem~\ref{theo:second_result}.
In case we need to to pre-process the input graph in advance, the run-time complexity will remain the same -- see Appendix~\ref{app:adjacency_lists}.

\paragraph{Clustering with the normalized Laplacian}

As we will show in Section~\ref{sec:approx-adjacency-matrix}, the equivalent of Eq.~(\ref{eq:lapl-psd-approx}) for normalized Laplacians does not hold in general. This means that, in principle, we cannot apply spectral clustering to the normalized Laplacian for the approximate motif graph and keep the same theoretical guarantees. However, we can still make use of the approximate adjacency matrix if, somehow, we happen to know the motif \emph{degrees} exactly -- see Section~\ref{sec:approx-norm-laplacian} for details. Unfortunately, computing the motif degrees exactly takes as much time as doing a Grover search over all motif instances, and then the run-time is the same as the run-time of Algorithm~\ref{algo:motif_grover}. 

Nevertheless, in the absence of a firm theoretical footing, the numerical simulations presented in Section~\ref{sec:numerics-norm-laplacian} suggest that, in practice, we can use the approximate motif graph and the corresponding approximate motif degrees to cluster successfully using the normalized Laplacian.

\subsubsection{Motif clustering with quantum counting and quantum clustering}
\label{sec:clustering_qcount_quantum_spectral_clustering}
In the case that $d^l = \omega(\sqrt{n})$, it is possible to obtain a more efficient algorithm by not constructing the entire motif graph beforehand, but instead providing query access to it. To do this, we can assume that the motif graph is fully connected, but that non-edges have weight $0$. Then, using Algorithm~\ref{algo:motif_count} of Lemma~\ref{lem:approx_adjacencymatrix} with a constant $\epsilon$, we can provide adjacency list access (which is now equivalent to adjacency matrix access) to the motif graph $\fG$ using $\tilde{O}(\sqrt{d^{s-2}})$ queries to the input graph $G$ and $\tilde{O}(\sqrt{d^{s-2}})$ other operations, and then directly use the quantum spectral clustering algorithm of Apers and de Wolf from Corollary~\ref{cor:qsc_apers_dewolf} to cluster, which will require $\tilde{O}(n^{3/2})$ queries to the adjacency lists of $\fG$. This will yield an algorithm for performing motif clustering that takes
\[
    \tilde{O}\left(n^{3/2}\sqrt{d^{s-2}}  + \Tk\right)
\]
time, which proves Theorem~\ref{theo:third_result}. As before, pre-processing the input graph in advance does not affect the run-time -- see Appendix~\ref{app:adjacency_lists}. The process described above is encapsulated in Algorithm~\ref{algo:motif_cluster_qcount_qcluster}.

\begin{algorithm}[H]
	
	\caption{Quantum motif clustering via quantum approximate counting and quantum spectral clustering}
	
	 \hspace*{\algorithmicindent} \textbf{Input}: A graph $G=(V,E)$, motif $M$ with two anchor nodes and integer $k\geq 1$. 

	\begin{algorithmic}[1]

		\State Use Algorithm \ref{algo:motif_count} of Lemma~\ref{lem:approx_adjacencymatrix} on input $G=(V,E)$, motif $M$, some fixed relative error $\epsilon>0$, and with failure probability at most $\delta = O(1/n^{3/2})$ to provide adjacency list access to the (all-to-all) approximate motif graph $\tilde{\fG}$.
		\State Apply the algorithm from Corollary~\ref{cor:qsc_apers_dewolf} to the unnormalized motif Laplacian $\tilde{\fL}$ by providing query access to $\tilde{\fG}$.
		\State \Return A $k$-partition of $G=(V,E)$ that approximately minimizes the motif RatioCut. 
	\end{algorithmic}
	\label{algo:motif_cluster_qcount_qcluster}
\end{algorithm}

\noindent We note that, by providing query access to the motif graph $\fG$ rather than constructing it explicitly, we lose the ability to detect and remove isolated vertices in $\fG$. These will now be included in the graph provided as input to the clustering subroutine, which will almost certainly assign each isolated vertex to its own cluster. This may impact the quality of the solutions found by the motif clustering algorithm. However, as discussed in Section~\ref{sec:results}, we note that this version of the algorithm should only be used if $d^l$ and the total number of motif instance $\fM$ are relatively large, in particular $d^l = \omega(\sqrt{n})$ and $\fM = \omega(n^2/d)$. In this case, the input graph is quite dense, and also there are reasonably many motif matches, making it not unlikely that the number of isolated vertices in $\fG$ will be quite small. 

Finally, we should note that for Algorithm~\ref{algo:motif_cluster_qcount_qcluster}, we can in general not use the normalized Laplacian for clustering, since constructing it could mean that we are dividing by zero due to some vertices possibly being disconnected.

\subsection{Run-time comparisons}
\label{sec:comparison_run-times}

Next, we discuss how the run-time $O(nd^{s-1})$ of the classical algorithm compares to the run-times of the quantum algorithms introduced in this section. We will ignore the time it it takes to do $k$-means, which is the same for all algorithms considered, and we assume is nearly-linear in $n$.

In order to investigate the run-time of our Grover-based Algorithm~\ref{algo:motif_cluster_grover}, let us take the more general starting point and assume that the adjacency lists are initially not sorted. In this case, we can either pre-process the input graph (which includes loading the input graph to QRAM if needed), or we can use Grover search to search through the adjacency lists (assuming we have coherent access to the input graph). 

The run-time of Algorithm~\ref{algo:motif_cluster_grover} depends on the number of motif instances $\fM$. If we pre-process the input graph, we get a speedup over the classical algorithm when $\fM = o(nd^{s-1})$, but this speedup is limited by the time it takes to do the sorting. If $\fM = \Omega(nd^{s-1})$, there is no speedup at all over the classical algorithm. In case of coherent access to unsorted adjacency lists of the input graph, we can also choose not to pre-sort the adjacency lists, but this only makes sense if the time it takes to run  Algorithm~\ref{algo:motif_cluster_grover} without pre-sorting is faster than the time it takes to sort the adjacency lists. Comparing the upper bound of $\tilde{O}(\sqrt{nd^s\fM})$ for the former to the upper bound of $O(nd)$ for the latter, we conclude that not pre-sorting is favorable if $\fM = o\left(\frac{n}{d^{s-2}} \right)$, albeit that what will be best in practice will depend on how tight these upper bounds are.

The quantum algorithms that use quantum approximate counting are a bit easier to analyse because they do not depend on $\fM$, nor on whether the adjacency lists of the input graph are sorted or not, or if we have classical or coherent access to them, as we can always pre-process the input graph beforehand. The run-time of Algorithm~\ref{algo:motif_cluster_qcount} is $\tilde{O}(nd^{l+\frac{s}{2}-1})$, which for motifs with $2l<s$ provides a speedup over the classical algorithm. The run-time of Algorithm~\ref{algo:motif_cluster_qcount_qcluster} is $\tilde{O}\left(n^{3/2}d^{\frac{s}{2}-1} \right)$, providing a speedup over the classical algorithm if $\sqrt{nd^s} = o(d^{s})$, which will be the case for, for example, dense graphs.

\subsubsection{Clique motifs in scale-free networks}
Let us investigate the run-times of the quantum motif clustering algorithms for a class of network that occurs often in practice: so-called scale-free networks ~\cite{albert1999diameter, vazquez2002large, barabasi2003scale, prvzulj2007biological, lima2009powerful}. Such networks have degree distributions that can be well approximated by a power-law distribution so that the fraction of vertices of degree $h$ scales as $h^{-\tau}$ for some $\tau>1$. 


As an example, consider motif clustering with the motif being an $s$-clique with two anchor nodes. We take $s \geq 3$ to be a constant independent of $n$. For clique motifs the distance between the two anchor nodes is $l=1$. 

In \cite{janssen2019counting}, the authors consider the number of $s$-cliques present in power-law random graphs on $n$ vertices with parameter $\tau\in (2,3)$ (they consider the so-called ``hidden variable'' model, see \cite{chung2002average, boguna2003class, britton2006generating, bollobas2007phase}) and find that its expected value is given by $O(n^{\frac{s}{2}(3-\tau)})$ as $n\to \infty$. This means that for such graphs $\mathbb{E}[\fM] = O(n^{\frac{s}{2}(3-\tau)})$, where the expectation is taken over the randomness of the input graph. The maximum degree in these power-law random graphs is given by $d = n^{1/(\tau-1)}$ \cite{boguna2003class} (also known as the ``natural cut-off"). 
 
\paragraph{Run-time upper bounds}
Using the above upper bound for $d$ together with Jensen's inequality, we obtain the following upper bound for the expected number of queries to the graph for Algorithm \ref{algo:motif_cluster_grover} without pre-sorting the adjacency lists:
\[
\mathbb{E}\left[\tilde{O}\left(\sqrt{nd^{s}\fM}\right) \right] = \tilde{O}\left(\sqrt{n^{1 + \frac{s}{\tau - 1}}\mathbb{E}[\fM]}\right)=\tilde{O}\left(n^{\frac{1}{2}+\frac{s}{2}(\frac{1}{\tau-1}+\frac{3-\tau}{2})}\right).
\]
Similarly, we can obtain upper bounds to the run-times of the other quantum algorithms. These are given in Table~\ref{tab:run_time_clique}.
\begin{table}[H]
\centering
\begin{tabular}{|l|l|l|l|}
\hline
    \textbf{Algorithm} & \textbf{Expected run-time}
    \\ \hline
\multicolumn{1}{|l|}{Classical} 
    & $O(n^{1+\frac{s-1}{\tau-1}})$ 
    \\ \hline
\multicolumn{1}{|l|}{Algorithm~\ref{algo:motif_cluster_grover}: quantum-Grover (pre-process)} 
    & $\tilde{O}\left(n^{\frac{\tau}{\tau -1}}+n^{\frac{1}{2}+\frac{s}{2}(\frac{1}{\tau-1}+\frac{3-\tau}{2}) - \frac{1}{2(\tau-1)}}\right)$ 
    \\ \hline
\multicolumn{1}{|l|}{Algorithm~\ref{algo:motif_cluster_grover}: quantum-Grover (no pre-process)} 
    & $\tilde{O}(n^{\frac{1}{2}+\frac{s}{2}(\frac{1}{\tau-1}+\frac{3-\tau}{2})})$ 
    \\ \hline
\multicolumn{1}{|l|}{Algorithm~\ref{algo:motif_cluster_qcount}: quantum-approximate + classical cluster} 
    & $\tilde{O}(n^{1+\frac{s}{2(\tau-1)}})$ 
    \\ \hline 
\multicolumn{1}{|l|}{Algorithm~\ref{algo:motif_cluster_qcount_qcluster}: quantum-approximate + quantum cluster}  
    & $\tilde{O}(n^{\frac{3}{2}+\frac{s}{2(\tau-1)} - \frac{1}{\tau -1}})$ 
    \\ \hline
\end{tabular}%
\caption{Expected run-times of motif clustering algorithms when the $n$-vertex input graph is a power law graph with $\tau \in (2,3)$, and the motif is a clique of constant size $s$.}
\label{tab:run_time_clique}
\end{table}

\paragraph{Comparison of run-time upper bounds}
The above (upper bounds for the)\footnote{For the remainder of this section, to make it easier to read, we will omit the word ``upper bound'', and keep in mind that the word ``run-time'' actually refers to an upper bound to the run-time.} run-times look somewhat complicated. Let's compare them to each other to see which algorithm has the fastest run-time depending on the choice of $s$ and $\tau \in (2,3)$. Intuitively, we expect that the Grover-based algorithms will perform well when there are few motif instances --- i.e. when $\mathbb{E}(\fM) = O(n^{\frac{s}{2}(3 - \tau)})$ is small --- which is the case for $\tau$ close to 3. As $\tau$ decreases, the graph becomes denser (since $d = n^{\frac{1}{\tau -1}} \rightarrow n$ as $\tau \rightarrow 2$), and we expect the quantum approximate counting based algorithms to do better. This intuition turns out to be correct.

First, observe that, since $\tau \in (2,3)$, we have $1 \geq \frac{3}{2} - \frac{1}{\tau -1}$, and therefore \textsc{quantum-approximate + quantum cluster} is faster than \textsc{quantum-approximate + classical cluster}, and so we should always use the former given a choice between the two. Second, comparing the two versions of quantum-Grover, we observe that we should only \emph{not} pre-sort if 
\[
    \frac{1}{2}+\frac{s}{2}\left(\frac{1}{\tau-1}+\frac{3-\tau}{2}\right) \leq \frac{\tau}{\tau -1}\, ,
\]
which is true only  when $s=3$, and $\tau \in [\tau_1,3)$, where $\tau_1 = \frac{5 + \sqrt{10}}{3} \approx 2.72$. In particular, for $s \geq 4$, we should always pre-sort the adjacency lists given the choice between both quantum-Grover algorithms.

Next, we compare the run-time upper bounds of the competing algorithms for different values of $s$. For $s \geq 4$, the only two competing algorithms are \textsc{quantum-approximate + quantum cluster} and \textsc{quantum-Grover (pre-process)}. Note that, in this regime for $s$,
\[
    \frac{3}{2}+\frac{s}{2(\tau-1)} - \frac{1}{\tau -1} \geq \frac{\tau}{\tau -1 }\, ,
\]
and therefore \textsc{quantum-approximate + quantum cluster} is always slower than the time it takes to pre-sort the adjacency lists. Hence, all that remains is to compare the second term in the run-time of \textsc{quantum-Grover (pre-process)} to the run-time of \textsc{quantum-approximate + quantum cluster}. After simplifying a bit, the intersection of the two run-times can be found by solving
\[
    \frac{s}{4}(3-\tau) = 1 - \frac{1}{2} \frac{1}{\tau -1}
\]
for $\tau \in (2,3)$ in terms of $s$. The result is 
\[
    \tau_0(s) = \frac{\sqrt{s^2 -2s +4} + 2s - 2}{s},
\]
which lies in the interval $(2,3)$ for $s\geq 3$. 

For the case $s=3$, we first compare \textsc{quantum-Grover (pre-process)} to \textsc{quantum-approximate + quantum cluster} for $\tau \in (2,\tau_1)$. In this interval, we want to use \textsc{quantum-approximate + quantum cluster} for $\tau \in (2,\tau_0)$ and \textsc{quantum-Grover (pre-process)} for $\tau \in (\tau_0, \tau_1)$, where $\tau_0 = \tau_0(s=3) = \frac{4 + \sqrt{7}}{3} \approx 2.22$. Finally, for $\tau \in (\tau_1, 3)$, \textsc{quantum-Grover (no pre-process)} has a faster run-time than \textsc{quantum approximate + quantum cluster}. 

\paragraph{Run-times of fastest algorithms}
In short, out of all quantum algorithms presented, we obtain the smallest run-time for \textsc{quantum-approximate + quantum cluster} for $\tau \in (2,\tau_0(s))$. The run-time for \textsc{quantum-Grover (pre-process)} is smallest for $\tau \in (\tau_0(s),3)$, except for $s=3$, in which case \textsc{quantum-Grover (no pre-process)} has an even smaller run-time for $\tau \in (\tau_1, 3)$. 

As an example, in Table~\ref{tab:run_time_s34} we list the explicit upper bounds to the run-times for the fastest algorithms as given above for the cases of $s=3$, $s=4$, and $s=5$, in the limits of $\tau \rightarrow 2$ and $\tau \rightarrow 3$, and compare them to the run-time of the classical algorithm, (using the upper bound for the maximum degree given by the natural cut-off $d = n^{\frac{1}{\tau -1}}$).
\begin{table}[H]
\centering
\begin{tabular}{l l|l|l l}

& & \textbf{Classical} & \textbf{Quantum} \\ \hline
& & run-time &  run-time & fastest algorithm  \\ \hline
$s=3$ & $\tau \rightarrow 2$ & $O(n^3)$ & $\tilde{O}(n^2)$ & quantum-approximate + quantum cluster\\
    & $\tau \rightarrow 3 $ & $O(n^2)$ &  $\tilde{O}(n^{1.25})$ & quantum-Grover (no pre-process)\\
    \hline
$s=4$ & $\tau \rightarrow 2$ & $O(n^4)$ & $\tilde{O}(n^{2.5})$ & quantum-approximate + quantum cluster\\
    & $\tau \rightarrow 3$ & $O(n^{2.5})$ & $\tilde{O}(n^{1.5})$ & quantum-Grover (pre-process)\\
    \hline
$s=5$ & $\tau \rightarrow 2$ & $O(n^5)$ & $\tilde{O}(n^{3})$ & quantum-approximate + quantum cluster\\
    & $\tau \rightarrow 3$ & $O(n^{3})$ & $\tilde{O}(n^{1.5})$ & quantum-Grover (pre-process)\\
\end{tabular}%
\caption{Run-times of the classical and fastest quantum motif clustering algorithms in expectation over the randomness of the $n$-vertex input power-law graph, where the motif is a clique of size $s \in \{3,4,5\}$. If we need to pre-process because our access to the input graph is classical, then for $s=3$ and $\tau \rightarrow 3$ the quantum run-time becomes $O(n^{1.5})$}
\label{tab:run_time_s34}
\end{table}
\noindent Note that we get a quadratic speedup on the run-time of the entire algorithm for the case $s=5, \tau \rightarrow 3$. This happens because, for $\tau \rightarrow 3$, $\fM \rightarrow O(1)$, and pre-processing the input graph takes as much time as the Grover search: both are given by $\tilde{O}(n^{1.5})$.


\section{Clustering on an approximate graph}
\label{sec:approx-adjacency-matrix}

In this section we provide analytical and numerical evidence to support the claim that performing spectral clustering using the Laplacian or the normalized Laplacian, respectively, on an approximation of a graph yields similar clusters to those that would be obtained by clustering on the actual graph. This is of particular relevance to us since our use of quantum approximate counting in Algorithms~\ref{algo:motif_cluster_qcount} and~\ref{algo:motif_cluster_qcount_qcluster} produces only an approximation of the motif graph, with each edge weight approximated up to some small constant relative error. 

We will consider the reasonably general case in which we wish to perform spectral clustering on some graph $G = (V,E)$ with real-valued adjacency matrix $A$ with non-negative coefficients, but where we only have access to an $\epsilon$-approximation $\tilde{A}$ of $A$, in the sense that $(1 - \epsilon) A_{uv} \leq \tilde{A}_{uv} \leq (1 + \epsilon) A_{uv}$ for all $u,v \in V$. Note that $A_{uv} =0$ if and only if $\tilde{A}_{uv} = 0$, so $A$ and $\tilde{A}$ have the same edge set $E$ (as long as $0 < \epsilon < 1$). Our central question is whether performing spectral clustering on the approximate graph yields clusters that are similar to the clusters obtained by performing spectral clustering on $G$ itself.

\subsection{Approximating the unnormalized Laplacian}

If we choose to cluster using the (unnormalized) Laplacian, then it turns out that we can answer the question above positively. That is: if we perturb the weights on the edges of a weighted graph by adding a multiplicative error, the spectrum of the Laplacian is preserved up to a similar multiplicative error. 

To make this precise, we first introduce some extra notation. Let $G = (V,E)$ be a undirected graph with symmetric real-valued adjacency matrix $A$ with non-negative coefficients. For a vertex $u \in V$, we define the indicator vector $1_u$ to be:
\[
    1_u(v) := 
    \begin{cases}
    1 & \text{if} \quad v = u, \\
    0 & \text{otherwise}.
    \end{cases}
\]
Now, can write the Laplacian $L = D - A$, where $D$ is the weighted degree matrix of $A$ (with diagonal coefficients $D_{uu} = \sum_{v \in V} A_{uv}$ for every $u \in V$), as 
\[
    L = \sum_{\{u,v\} \in E} A_{uv} L_{\{u,v\}},
\]
where for each edge $\{u,v\} \in E$, $L_{\{u,v\}} := (1_u - 1_v)(1_u - 1_v)^T$. Then we have the following result.

\begin{lemma}\label{lem:psd-relation_laplacians}
Let $A$ and $\tilde{A}$ both be symmetric adjacency matrices of a graph $G=(V,E)$ with real-valued non-negative weights. Let $D$ and $\tilde{D}$ be the corresponding weighted degree matrices, and $L = D - A$ and $\tilde{L} = \tilde{D} - \tilde{A}$ the corresponding Laplacians, and let $\epsilon > 0$. Then, if for all $u,v\in V$  we have $(1 - \epsilon) A_{uv} \leq \tilde{A}_{uv} \leq (1 + \epsilon) A_{uv}$, it follows that
\begin{equation}
    (1 - \epsilon) L \preceq \tilde{L} \preceq (1 + \epsilon) L.
    \label{eq:lapl_psd_lemma}
\end{equation}
\end{lemma}

\begin{proof}
We need to show that both $\tilde{L} - L + \epsilon L$ and $L - \tilde{L} + \epsilon L$ are positive semi-definite. The inequalities $(1 - \epsilon) A_{uv} \leq \tilde{A}_{uv} \leq (1 + \epsilon) A_{uv}$ for every $u,v \in V$ imply that there exists a matrix $\gamma$ such that $\tilde{A}_{uv} - A_{uv} = \epsilon A_{uv} \gamma_{uv}$ with $\gamma_{uv} \in [-1, 1]$ for all $u,v \in V$. Hence,
\[
    \tilde{L} - L + \epsilon L = \sum_{(u,v) \in E} (\tilde{A}_{uv} - A_{uv} + \epsilon A_{uv}) L_{\{u,v\}} 
    = \sum_{(u,v)\in E}  \epsilon A_{uv} (\gamma_{uv} + 1) L_{\{u,v\}},
\]
which is the Laplacian of a graph $G = (V,E)$ where each edge $(u,v) \in E$ has weight $\epsilon A_{uv} (\gamma_{uv} + 1) \geq 0$, and hence is itself a positive semi-definite matrix. 
Likewise,
\[
    L - \tilde{L} + \epsilon L = \sum_{(u,v)\in E} \epsilon A_{uv} (1 - \gamma_{uv}) L_{\{u,v\}},
\]
is also the Laplacian of a graph $G = (V,E)$ where each edge $(u,v) \in E$ has weight $\epsilon A_{uv} (1 - \gamma_{uv}) \geq 0$, and therefore is also positive semi-definite.
\end{proof}

\subsection{Approximating the normalized Laplacian}
\label{sec:approx-norm-laplacian}

A natural question to ask is if a similar result to that of Lemma~\ref{lem:psd-relation_laplacians} also holds for the normalized Laplacian. That is, does the statement: for every $\delta >0$ there exist a $\epsilon > 0$, such that if $A$ and $\tilde{A}$ satisfy $(1 - \epsilon) A_{uv} \leq \tilde{A}_{uv} \leq (1 + \epsilon) A_{uv}$ for every $u,v \in V$, with corresponding Laplacian $L = D - A$, approximate laplacian $\tilde{L} = \tilde{D} - \tilde{A}$ and corresponding normalized Laplacians given by
\begin{equation}
     L_{\n} = D^{-\frac{1}{2}} L D^{-\frac{1}{2}} \quad \text{and} \quad  \tilde{L}_{\n} = \tilde{D}^{-\frac{1}{2}} \tilde{L} \tilde{D}^{-\frac{1}{2}},
     \label{eq:norm_lapl_and_perturbed_norm_lapl}
\end{equation}
then
\begin{equation}
    (1 - \delta) L_{\n} \preceq \tilde{L}_{\n} \preceq (1 + \delta) L_{\n} \,;
    \label{eq:norm-lapl-false}
\end{equation}
hold? (Note that for Lemma~\ref{lem:psd-relation_laplacians} we have $\epsilon = \delta$.)

First, we observe that, if we choose to normalize the approximate Laplacian with the true weighted degree matrix $D^{-\frac{1}{2}}$ rather than the approximate weighted degree matrix $\tilde{D}^{-\frac{1}{2}}$, then the above statement does hold (for $\epsilon = \delta$).

\begin{lemma}
Let $A$ and $\tilde{A}$ be  adjacency matrices of a graph $G=(V,E)$ with real-valued non-negative weights, and let $\epsilon > 0$. Now, if $(1 - \epsilon) A_{uv} \leq \tilde{A}_{uv} \leq (1 + \epsilon) A_{uv}$ holds for all $u,v\in V$, then the normalized Laplacian $L_{\n} = D^{-\frac{1}{2}}L D^{-\frac{1}{2}}$ and the approximate Laplacian normalized with true degree matrix $L_{\n}' = D^{-\frac{1}{2}}\tilde{L}D^{-\frac{1}{2}}$, where $L = D - A$, and $\tilde{L} = \tilde{D} - \tilde{A}$, are related by
\begin{equation}
    (1 - \epsilon) L_{\n} \preceq L_{\n}' \preceq (1 + \epsilon) L_{\n}.
    \label{eq:norm-lapl-psd-approx}
\end{equation}
\end{lemma}

\begin{proof}
Eq.~\eqref{eq:norm-lapl-psd-approx} follows directly from Eq.~\eqref{eq:lapl_psd_lemma}, and the fact that, if $X \succeq 0$, then $X = B^\dagger B$ for some matrix $B$, and therefore $D^{-\frac{1}{2}} X D^{-\frac{1}{2}} =  \left(B D^{-\frac{1}{2}} \right)^\dagger BD^{-\frac{1}{2}} \succeq 0$. 
\end{proof}

As a consequence, for graphs for which we know the motif degrees exactly, we can use quantum counting to compute the matrix $L_{\n}'$, which in turn can be used for $k$-means spectral clustering with $\epsilon$-approximate eigenvectors. Unfortunately, if we do not know the motif degrees exactly, then quantum counting does not offer any additional benefit over Algorithm~\ref{algo:motif_grover}, since using quantum counting to count all the degrees exactly has the same complexity as finding all motif instances.

\ 
\noindent Next, we will show that Eq.~\eqref{eq:norm-lapl-false} does not hold in general. Specifically, Eq.~\eqref{eq:norm-lapl-false} does not hold unless, coincidentally, all perturbations are such that $D  = c \tilde{D}$ for some real-valued constant $c$. Consequently, we do not have the same guarantees for the approximate normalized Laplacian as we do for the approximate (unnormalized) Laplacian.

To this end, let $1 > \delta > 0$. We will show that, no matter how small we pick $\epsilon > 0$, if $\Ln$ are $\Lnt$ as described in and above Eq.~\eqref{eq:norm_lapl_and_perturbed_norm_lapl}, we do not have that
\begin{equation}
    (1 - \delta) \Ln \preceq \Lnt \preceq (1 + \delta) \Ln\, .
\end{equation}
In particular, we will show that
\begin{equation}
    \Lnt - (1 - \delta) \Ln \succeq 0\
\end{equation}
does not hold regardless of how small we pick $\epsilon > 0$.

We know that $\vec{v} = \sqrt{D}\one$ is a 0-eigenvector of $\Ln$, and likewise $\vec{\tilde{v}} = \sqrt{\tilde{D}}\one$ a 0-eigenvector of $\Lnt$\footnote{$\one$ is a 0-eigenvector of the unnormalized Laplacian $L$, and then we have that $\Ln \sqrt{D}\one = D^{-1/2} L D^{-1/2} \sqrt{D}\one = 0$}. Now, let us first assume that the underlying graph for $A$ is connected. In that case, $\vec{v}$ is the only 0-eigenvector of $\Ln$, and all other eigevectors have an eigenvalue that is stricly positive. Consequently, as long as $\vec{v}$ and $\vec{\tilde{v}}$ are not linearly dependent, which happens when the degree matrices $D$ and $\tilde{D}$ are not constant multiples of each other, we have
\[
    \vec{\tilde{v}}^T \left[\Lnt - (1-\delta)\Ln \right] \vec{\tilde{v}} = - (1-\delta) \vec{\tilde{v}}^T  \Ln \vec{\tilde{v}} < 0\,.
\]
This implies that the first inequality above does not hold for general $\epsilon > 0$, regardless of how small we choose $\delta$, unless all perturbations are such that $D  = c \tilde{D}$ for some constant $c>0$. 

If the graph is not connected, we can just restrict ourselves to a single connected component of $A$, and repeat the argument for that connected component. 

\ 

\noindent The above observation essentially rules out obtaining strong guarantees in the case of clustering on approximate normalized Laplacians. However, in the next section we provide numerical evidence to suggest that, in fact, we can actually use the approximate normalized Laplacian for clustering in practice.

\subsection{Numerical simulations for the approximate normalized Laplacian}
\label{sec:numerics-norm-laplacian}
Even though we cannot obtain theoretical guarantees in the case of clustering via the normalized Laplacian when the weights of the graph are only known approximately, we give evidence to suggest that, in practice, the situation is similar to the unnormalized case. Recall that in the latter, the spectrum of the graph Laplacian is preserved up to small constant multiplicative error when the weights on the edges of the graph are perturbed by a similarly small constant multiplicative error (see Lemma~\ref{lem:psd-relation_laplacians}), and thus spectral clustering will perform similarly on the original and the perturbed graph.

In what follows, we study empirically what happens to the quality of the clusters produced by spectral clustering on the normalized Laplacian when we perturb the edges of weighted, randomly generated graphs. In particular, we consider weighted undirected graphs $G=(V,E)$ with edge weights $\{w_e | e\in E\}$, and their `perturbed' versions with edge weights $\{w'_e | e \in E\}$, where each $w'_e$ is drawn uniformly at random from $[(1-\epsilon)w_e, (1+\epsilon)w_e]$ for some relative error $\epsilon \in [0,1]$. 

We find that, in general, only large values of $\epsilon$ yield significant differences in the quality of clusters obtained by spectral clustering applied to the normalized Laplacian of the graph. As an illustration, consider the graph shown in Figure~\ref{fig:circle_clusters}, a commonly used test-case for demonstrations of clustering algorithms, which consists of two concentric circles of points embedded in $\mathbb{R}^2$. Here we added edges between nearby points,\footnote{More precisely, we generated data points in $\mathbb{R}^2$ representing two concentric circles using the \textit{Scikit-learn} Python library, before scaling them to remove the mean (i.e. set it to zero) and obtain unit variance. Edges were then added between points $\mathbf{x}$ and $\mathbf{y}$ if their Euclidean distance satisfied $d(\mathbf{x},\mathbf{y}) 
\leq 0.6$, and given weight $\frac{20}{d(\mathbf{x},\mathbf{y})}$.} with a weight that scales inversely proportional to their Euclidean distance, and then applied spectral clustering to the normalized Laplacian of the resulting graph. Only after introducing a relative error of $\epsilon=0.6$ did we find that the resulting clusters differed at all from those found on the original graph. We note that this graph was even handpicked to demonstrate that perturbations can qualitatively change the clusters obtained from spectral clustering -- in fact most graphs we generated were much more resilient to perturbations! This is almost certainly due to the fact that these graphs are `well clusterable' in the sense that the two clusters are easily (albeit non-linearly) separable in a low-dimensional space. 

\begin{figure}
    \hspace{-1cm}
    \includegraphics[scale=0.6]{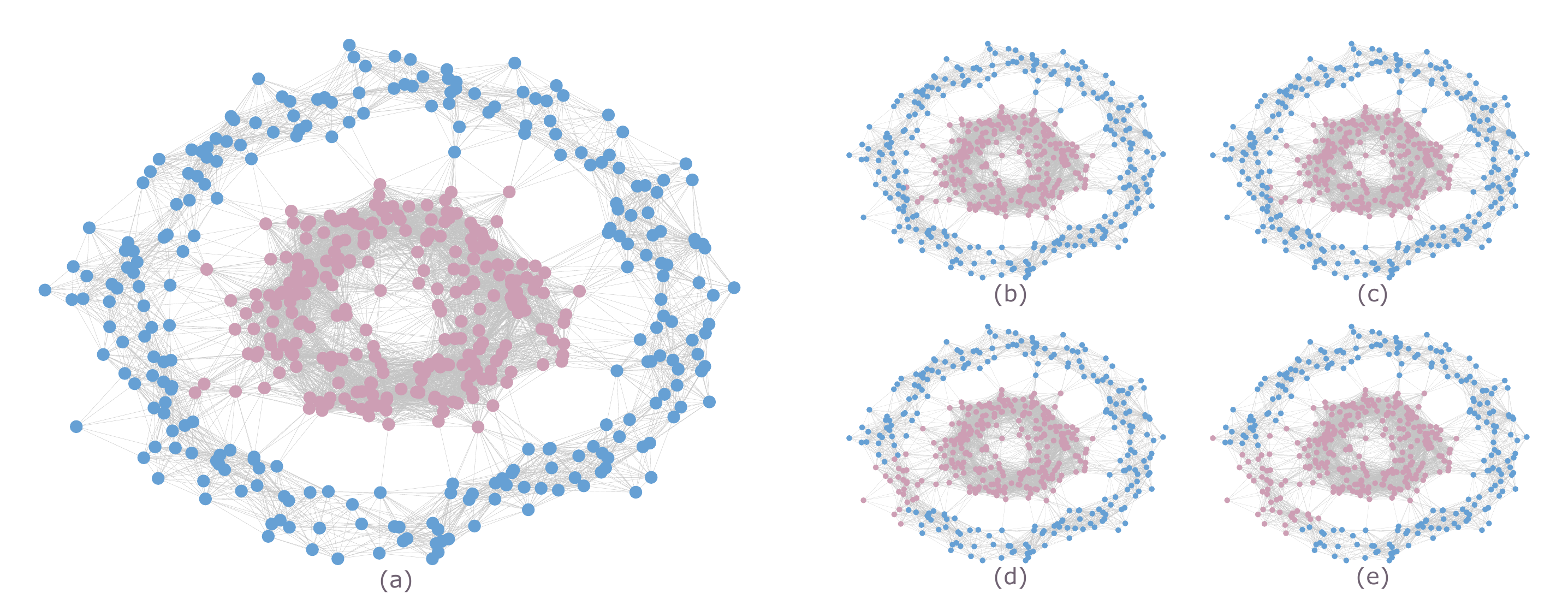}
    \caption{Effect of applying relative errors to the weights on the edges of a graph and then applying spectral clustering ($k=2$) to the resulting normalized Laplacian. (a) Original graph. (b) $\epsilon=0.6$. (c) $\epsilon=0.75$. (d) $\epsilon=0.9$. (e) $\epsilon=1$.}
    \label{fig:circle_clusters}
\end{figure}

To more precisely quantify the effect of relative errors on graph clustering algorithms, we consider the difference in conductance (see Section~\ref{sec:motif_clustering}) achieved by spectral clustering on the original and perturbed graphs. More precisely, let $\{W_i\}_{i=1}^k$ be the partition (i.e. set of clusters) output by applying spectral clustering to the original graph, and $\{\tilde{W}_i\}_{i=1}^k$ the partition found by applying it to the perturbed graph. Then we use the quantity
\[
    \phi_{\text{diff}} := \phi_G(\tilde{W}_1,\dots,\tilde{W}_k) - \phi_G(W_1,\dots,W_k)
\]
as a measure of the difference in quality between the two partitions. Note that the conductance for both partitions is computed relative to the \emph{original} graph $G$ (i.e. without perturbed weights). Since spectral clustering aims to minimize the conductance, this is the natural quantity to capture the difference in quality of two different partitions of the same graph. If the partitions output by spectral clustering on the perturbed graph are worse, then $\phi_{\text{diff}}$ will be positive; if they happen to be better, it will be negative.

\begin{table}[]
\centering
\begin{tabular}{|l|c|c|}
\hline
$n$ & Cluster graph & LFR \\ \hline
$600$ & $-0.001 \pm 0.019$ & $0.004 \pm 0.03$ \\ \hline
$800$ & $0.0 \pm 0.021$ & $-0.0 \pm 0.027$ \\ \hline
$1000$ & $-0.001 \pm 0.004$ & $0.001 \pm 0.031$ \\ \hline
$1200$ & $0.001 \pm 0.013$ & $-0.0 \pm 0.024$ \\ \hline
$1400$ & $0.001 \pm 0.016$ & $-0.002 \pm 0.039$ \\ \hline
$1600$ & $0.0 \pm 0.001$ & $-0.0 \pm 0.027$ \\ \hline
$1800$ & $0.001 \pm 0.016$ & $-0.005 \pm 0.028$ \\ \hline
$2000$ & $-0.002 \pm 0.022$ & $0.003 \pm 0.036$ \\ \hline
\end{tabular}%
\caption{Average and standard deviation of $\phi_{\text{diff}}$ vs $n$ for 200 randomly generated graphs with fixed relative error $\epsilon = 0.1$.}\label{tab:error_vs_n}
\end{table}

As test cases, we consider two types of random graphs: `cluster' graphs, which are created by generating random points in $\mathbb{R}^2$ centred around some number of fixed centres, and then adding edges between nearby points with weights that scale inversely proportional to the Euclidean distance between them; and so-called LFR-graphs~\cite{lancichinetti2008benchmark}, a commonly used family of random graphs used to test clustering and community-detection algorithms. For all tests we set $k=5$ (and for the cluster graphs, generated data centred around $k=5$ fixed centres). Table~\ref{tab:error_vs_n} shows the average value of $\phi_{\text{diff}}$ for a range of sizes of randomly generated graphs of both types for a fixed relative error of $\epsilon=0.1$. We find that, regardless of graph size, the perturbation has essentially no effect on the quality of clusters found.

\begin{figure}[H]
    \hspace{-1cm}
    \includegraphics[scale=0.65]{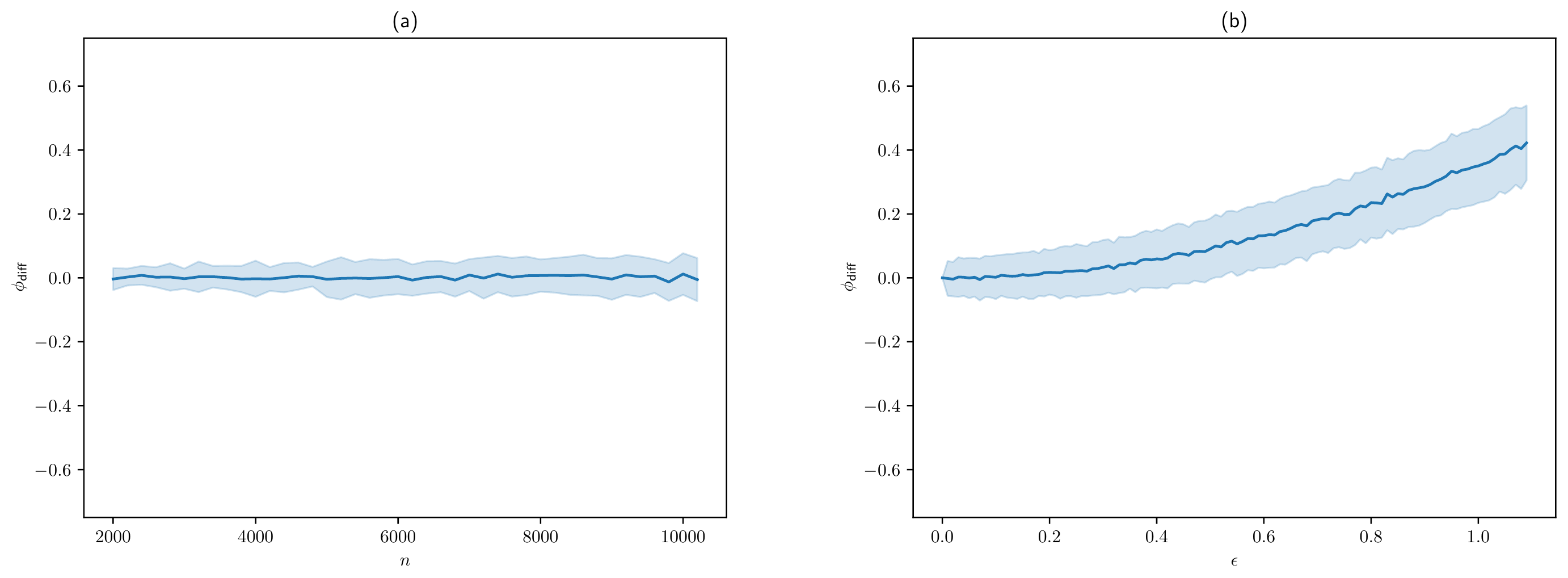}
    \caption{Effect of graph size and relative error on the quality of clusters found by spectral clustering on LFR graphs. (a) $\phi_{\text{diff}}$ vs. $n$ for fixed $\epsilon = 0.1$ with 200 random graphs per value of $n$; (b) $\phi_{\text{diff}}$ vs. relative error $\epsilon$ for fixed $n = 2000$ and 500 random graphs per value of $\epsilon$.}
    \label{fig:graph}
\end{figure}

Next we consider the effect of increasing relative error on the quality of clusters found by spectral clustering. Figure~\ref{fig:graph} shows, for randomly generated LFR graphs, the effect of increasing the graph size $n$ (left figure) and of increasing the relative error $\epsilon$ (right figure) on the value of $\phi_{\text{diff}}$. It is clear that the clusters do not become worse as $n$ increases (for fixed $\epsilon=0.1$), but do become worse as $\epsilon$ increases (here for fixed $n=2000$). This suggests that, as in the case of clustering using the unnormalized Laplacian, it suffices to choose a small, but \emph{constant} relative error $\epsilon$ to obtain good quality clusters via the approximate normalized Laplacian.

\section{Conclusion}
\label{sec:conclusion}

We have presented three quantum algorithms that provide a speedup over classical methods for performing motif clustering. Our speedup relies on quantum routines for finding or approximately counting the number of motif instances in the to-be-clustered graph. In the case of approximate quantum counting, we show that approximations up to only a constant relative error are sufficient for motif clustering using the unnormalized Laplacian of the motif graph, which allows us to obtain a quantum speedup in many cases.

This observation in fact holds more generally: if we perturb the weights of a graph with some constant relative error, then the graph can still be used to perform spectral clustering via the unnormalized Laplacian, which produces clusters whose RatioCut is close to the RatioCut that would have been obtained by performing spectral clustering on the unperturbed graph. It is interesting that the effect of the perturbation on the quality of the obtained clustering is independent of the size of the graph.

Our argument for the above claim fails, however, for normalized Laplacians: in this case, the spectrum of the Laplacian corresponding to the perturbed graph does not preserve the spectral structure of the unperturbed graph. However, when applied to randomly generated benchmark graphs, we find numerically that clustering using the normalized Laplacian of the perturbed graph does in fact generate clusters for which the conductance is close to the conductance of the clusters obtained by performing spectral clustering via the normalized Laplacian of the unperturbed graph, again independent of the size of the graph. An interesting open question would be to find out why, and under what conditions, clustering with the normalized Laplacian of the perturbed graph can be used to obtain a clustering with low conductance in the original graph.
 
In Appendix~\ref{sec:higher_order_motifs}, we discuss motif clustering of a graph $G$ using a motif with more than two anchor nodes. In particular, we show that clustering with a motif with three anchor nodes is equivalent to clustering with a weighted combination of two-anchor-node motifs. We continue to argue that, for motifs $M$ with four or more anchor nodes, a weighted combination of two-anchor-node motifs $M_1,\ldots, M_q$, corresponding to the motifs obtained by taking all possible pairs of anchor nodes of $M$, should be used in place of $M$ itself. The reason is that the motif graph $\fG$ constructed from $G$ and $M$ is itself a graph, and is therefore only capable of expressing pairwise relationships between vertices that occur as anchor nodes -- which is exactly what the motifs $M_1, \ldots, M_q$ represent -- and not higher-order relationships between sets of more than two vertices. If one wants the latter, then instead of constructing the motif graph $\fG$, one should construct a \emph{hypergraph} where the hyperedges represent the multi-vertex relationships expressed by $M$, and cluster on the hypergraph instead. 

There are existing classical algorithms for clustering on hypergraphs: see for example ~\cite{takai2020hypergraph} and references therein for a clustering method based on the so-called `personalized pagerank', which itself is based on the stationary distribution of a random walk on the hypergraph. An interesting open question would be to investigate if quantum walks can provide a speedup for the hypergraph clustering algorithm of~\cite{takai2020hypergraph}.

\paragraph{Acknowledgements}
We would like to thank Ian Marshall for his active participation in the early stages of the project, and for the many insights he had during our regular meetings. We would also like extend our gratitude to Simon Apers for his useful feedback and for pointing out that there are cases where not constructing the motif graph explicitly is the most efficient approach, as well as to Ronald de Wolf for reading an earlier version of this manuscript and providing feedback. In addition, we would like to thank Johan van Leeuwarden for clarifying a result in~\cite{janssen2019counting}, and Ton Poppe and Edo van Uitert at ABN AMRO for fruitful discussions on potential applications of quantum computing in transaction network analysis. Finally, we would like to thank an anonymous reviewer for pointing out a flaw in our proof of Lemma 11 and suggesting a fix.

\paragraph{Funding}
IN is a part of the DisQover project: a collaboration between QuSoft and ABN AMRO, and is partially funded by ABN AMRO. CC was supported by QuantERA project QuantAlgo 680-91-034. FL was supported by the Gravitation-grant NETWORKS-024.002.003 from the Dutch Research Council (NWO).

\appendix

\section{Motif isomorphisms}
\label{sec:motif_isomorphisms}

When using Algorithms~\ref{algo:motif_cluster_grover}, \ref{algo:motif_cluster_qcount} or~\ref{algo:motif_cluster_qcount_qcluster}, we count the number of tree walks with the property that the graph $G$ restricted to the vertex set of the constructed tree has the same edge structure as the motif $M$. Since each tree-walk is in one-to-one correspondence with a map $\iota: V_M \rightarrow V$, this means we are actually counting motif \emph{assignments}. However, when constructing the motif graph, given two vertices $u,v \in V$, we instead want to count the number of motif \emph{instances} that have $u$ and $v$ as anchor nodes, and therefore we must know how to obtain the latter from the former\footnote{See Section~\ref{sec:graph_motifs} for the definitions of motif assignments and motif instances, and for what it means for a motif instance to have a vertex as an anchor node.}.

As we will show in Lemma~\ref{lem:motif_isomorphisms} below, every motif instance corresponds to a constant number $S_M$ of motif assignments $\iota: V_M \rightarrow V$, and this constant $S_M$ only depends on the motif $M$ itself, and can be computed ahead of time. More precisely, $S_M$ is equal to the number of graph isomorphisms $f: V_M \rightarrow V_M$ of $M$ with the property that anchor nodes get mapped to anchor nodes: i.e. $f(V_A) = V_A$. It will be convenient to give such graph isomorphisms a name. In particular, for a motif $M$ with vertex set $V_M$ and anchor node set $V_A \subseteq V_M$, we say that a graph isomorphism $f: V_M \rightarrow V_M$ is a \emph{motif isomorphism} if $f$ maps anchor nodes to anchor nodes, i.e. $f(V_A) = V_A$. In Fig.~\ref{fig:motif_symmetries}, we compute $S_M$ for two example motifs.

\begin{figure}[!h]
\begin{center}
    \includegraphics[scale=1]{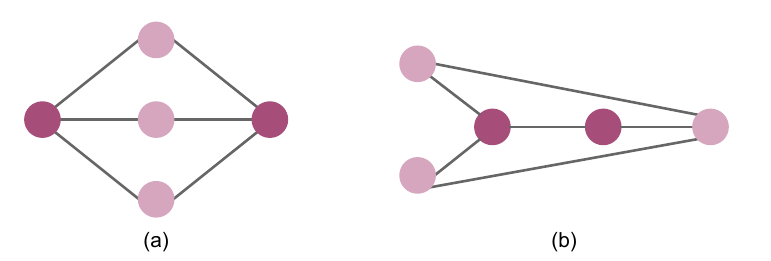}
\end{center}
\caption{Two example motifs. For motif (a), there are $S_M = 12$ motif isomorphisms, because there are six ways of interchanging the non-anchor nodes, and for each way of interchanging the non-anchor nodes we can choose to leave the anchor nodes invariant, or interchange them. For motif (b) the only graph isomorphism of the motif is the one that exchanges the left two non-anchor nodes, so $S_M = 2$.}
\label{fig:motif_symmetries}
\end{figure}

\begin{lemma}
\label{lem:motif_isomorphisms}
Let $G=(V,E)$ be an unweighted graph, $M=(V_M, V_A, E_M)$ a motif, and let $S_M$ be the number of motif isomorphisms of $M$. Then every motif instance of $G$ consists of exactly $S_M$ equivalent motif assignments $\iota:V_M \rightarrow V$.
\end{lemma}

\begin{proof}
We need to show that (i) any two motif assignments that correspond to the same instance are related by a motif isomorphism $f: V_M \rightarrow V_M$, and (ii) that for any motif assignment $\iota$, $\iota \circ f$ is equivalent to $\iota$ for any motif isomorphism $f: V_M \rightarrow V_M$. 

For (i), let $\iota$ and $\iota'$ be two motif assignments that correspond to the same motif instance, that is $\iota(V_M) = \iota'(V_M)$, and $\iota(V_A) = \iota'(V_A)$. Let $H$ be the subgraph of $G$ obtained by restricting $G$ to the image of $\iota$ (or $\iota'$; both have the same image). Now,  $f = \iota^{-1} \circ \iota': V_M \rightarrow V_M$ is a composition of graph isomorphisms, hence a graph isomorphism, and trivially we have $f(V_A) = V_A$. Consequently, $\iota' = \iota \circ f$, for some motif isomorphism $f$ of $M$. For (ii), given a motif isomorphism $f: V_M \rightarrow V_M$ and a motif assignment $\iota: V_M \rightarrow V$, observe that $\iota' = \iota \circ f: V_M \rightarrow V$ is a motif assignment, since it is a graph isomorphism onto its image, and also we have $\iota'(V_M) = \iota(f(V_M)) = \iota(V_M)$ as well as $\iota'(V_A) = \iota(f(V_A)) = \iota(V_A)$, so both $\iota$ and $\iota'$ correspond to the same motif instance.
\end{proof}

Consequently, when constructing the motif graph $\fG$, for each $u,v \in V$ we can use Algorithm~\ref{algo:motif_cluster_grover} to count the number of motif assignments $\iota:V_M \rightarrow V$ with $u,v \in \iota(V_A)$, and then divide by $S_M$ to obtain $\fA_{uv}$: the number of motif instances with $u$ and $v$ as anchor nodes. 

Similarly, if we are using Algorithms~\ref{algo:motif_cluster_qcount} or~\ref{algo:motif_cluster_qcount_qcluster}, which employ tree walks starting from two fixed anchor nodes (say $a$ and $b$) to count motif assignments, rather than dividing by $S_M$, we have to divide by the number $S_M^{(a,b)}$ of motif isomorphisms that leave the two anchor nodes $a$ and $b$ fixed. Since  Algorithms~\ref{algo:motif_cluster_qcount} or~\ref{algo:motif_cluster_qcount_qcluster} assume the motif has two anchor nodes (which will be $a$ and $b$), $S_M^{(a,b)}$ is the number of isomorphisms of $M$ that permute non-anchor nodes amongst each other.
Like $S_M$, $S_M^{(a,b)}$ only depends on the motif itself and can be computed in advance.

Note that the group of motif isomorphisms as well as the subgroup of motif isomorphisms that leave two given anchor nodes fixed are both subgroups of the permutation group on $s$ vertices. In particular, because we take $s$ to be constant, so are $S_M$ and $S_M^{(a,b)}$, and therefore the overhead coming from the fact that we are counting motif assignments rather than motif instances is a constant multiplicative overhead that does not affect the complexity of Algorithms~\ref{algo:motif_cluster_grover}, \ref{algo:motif_cluster_qcount} or~\ref{algo:motif_cluster_qcount_qcluster}.

\section{Motif graph cuts for two-anchor-node motifs}
\label{app:motif_graph_cuts}

Let $G = (V,E)$ be an unweighted graph with $|V| = n$, $W \subset V$ a subset of the vertex set $V$ of $G$, $M = (V_M, E_M, V_A)$ a motif with two anchor nodes: $V_A = \{a, b\}$, and $x: V \rightarrow \{-1,1\}^n$ the indicator function of $W$ that assumes the value 1 for vertices in $W$, and $-1$ for vertices outside $W$. If $\fG$ is the motif graph constructed from $G$ and $M$, then  $\text{cut}_{(G,M)}(W)$ is equal to $\text{cut}_{\fG}(W)$. 

Indeed, following the proof technique from~\cite{benson16}, we get the following equalities:
\begin{align*}
    \text{cut}_{(G,M)}(W) &= \frac{1}{4}  \sum_{\iota \in \mI} (x_{\iota(a)} - x_{\iota(b)})^2 \\
    &= \frac{1}{4} \sum_{\iota \in \mI} \left(  x_{\iota(a)}^2 + x_{\iota(b)}^2 - x_{\iota(a)} x_{\iota(b )} - x_{\iota(b)} x_{\iota(a)} \right) \\
    &=  \frac{1}{4} \left(x^T \fD x - x^T \fA x \right) \\
    &=  \frac{1}{4} x^T \fL x \\
    &= \text{cut}_{\fG}(W),
\end{align*}
where $\fA$, $\fD$ and $\fL = \fD - \fA$ are the adjacency matrix, degree matrix and Laplacian of $\fG$ respectively. Note that the factor 4 that appears in the one to last line to compensate for the factor $\frac{1}{4}$ arises because the indicator function $x$ takes values in $\{-1,1\}$ rather than $\{0,1\}$.

\section{Pre-processing the input graph}
\label{app:adjacency_lists}

All our algorithms require coherent access to the adjacency lists of the input graph. Moreover, 
the algorithms based on quantum approximate counting in Section~\ref{sec:motif_clustering_quantum_counting} assume that the adjacency lists of the input graph are sorted. If any of these conditions are not met, then we have to pay additional costs up front to pre-process the input graph, meaning that we have to either load the graph to QRAM, sort all adjacency lists, or do both. In this section we discuss what effect pre-processing the input graph has on the run-times of the algorithms presented in this work.

Let $G$ be the $n$-vertex input graph to which we have adjacency list access, with maximum degree $d$. Let $M$ be the $s$-vertex motif that we want to use to cluster $G$ with. If we have classical access to the input graph, then we first have to load it to QRAM in time $O(nd)$. If we have coherent access to the input graph, but the adjacency lists are not sorted, then we have two options: sort all adjacency lists beforehand in $\tilde{O}(nd)$ time, or keep the lists unsorted, in which case finding an element in an adjacency list takes $O(d)$ classically, or $\tilde{O}(\sqrt{d})$ using Grover search.

Recall that $s \geq 3$, as $s=2$ implies the motif is an edge, in which case there is no point in doing motif clustering. Also, the distance $l$ (number of edges) between any two anchor nodes satisfies $l\geq 1$.

\paragraph{Classical} The classical algorithm of Benson et al.~\cite{benson16} is happy to pay $\tilde{O}(nd)$ up front to sort the adjacency lists, as this always takes less time than the time it takes to find all motif instances. Hence, including pre-sorting the adjacency lists, the total run-time remains
\[
    \tilde{O}(nd^{s-1} + \Tk)
\]
for the entire classical $k$-means motif spectral clustering algorithm. 

\paragraph{Quantum via Grover search} 
Given classical access, we have to pre-process the input graph beforehand in time $\tilde{O}(nd)$, resulting in an expected run-time of 
\[
    \tilde{O}(nd + \sqrt{nd^{s-1} \fM} + \Tk)\, .
\]
If we are given coherent access, and the adjacency lists are sorted, then we have an expected run-time of
\[
    \tilde{O}(\sqrt{nd^{s-1} \fM} + \Tk)\, .
\]
If we have coherent access but the adjacency lists are unsorted, and we choose not to sort them in advance but use Grover search instead to check if and where a given node in a given motif instance occurs in the adjacency list of another node in the motif instance (at the cost of an extra factor of $\tilde{O}(\sqrt{d})$), then total expected run-time becomes
\[
    \tilde{O}(\sqrt{nd^{s} \fM} + \Tk)\, .
\]


\paragraph{Quantum via approximate counting and classical spectral clustering}
Since  $s \geq 3$ and $l \geq 1$, we have that $l  + \frac{s}{2} -1 > 1$, and therefore pre-processing the input graph in time $\tilde{O}(nd)$ does not affect the complexity of version of quantum motif clustering. The expected run-time of the algorithm including pre-processing remains
\[
    \tilde{O}(nd^{l + \frac{s}{2} -1} + \Tk).
\]

\paragraph{Quantum via approximate counting and quantum spectral clustering} As before, we can pre-process the input graph in time $\tilde{O}(nd)$ without changing the run-time of the algorithm: indeed, since $s\geq 3$, $nd \leq \sqrt{n^3 d} \leq \sqrt{n^3 d^{s-2}}$. Therefore, the expected run-time of this algorithm including pre-processing is given by
\[
    \tilde{O}\left(\sqrt{n^3 d^{s-2}} + \Tk \right).
\]

\section{Higher-order motifs}
\label{sec:higher_order_motifs}

Next, we consider more closely the role of anchor nodes. Recall from Section~\ref{sec:motif_cuts} that anchor nodes in the motif are the nodes that determine which graph cuts count as a motif cut and which do not. For motifs with two anchor nodes, the motif itself expresses a pairwise relation between both of its anchor nodes, and motif clustering can be thought of as clustering the original graph after applying some kind of filter to it: a filter that removes all connections except those that fit the motif pattern, see Fig~\ref{fig:motif_clustering}. 

We can also perform motif clustering for motifs with more than two anchor nodes, and as~\cite{benson16} show, doing so makes sense also for motifs with three anchor nodes. At first, it seems that a motif with three anchor nodes expresses relationships between three vertices. However, by construction the motif graph is still a graph, which captures only pairwise relationships. This begs the question: what is the interpretation of clustering using a motif with three anchor nodes? In the sections that follow, we address this question in detail.

\subsection{Multiple motifs}
\label{app:multiple_motifs}

Recall from the construction of the motif graph that we add $+1$ to the weight of every edge $(u,v)$ of the motif graph $\fG$ for every motif instance in $G$ with $u$ and $v$ as anchor nodes. This suggests that the motif graph obtained from a motif with multiple anchor nodes can be seen as a sum of two-anchor-node motif graphs: one for each pair of anchor nodes in the original motif. In the next subsection, we will argue that clustering using a motif with three anchor nodes is equivalent to clustering based on a combination of motifs with two anchor nodes. Before we can make this statement precise, we first follow the supplementary material of Benson et al.~\cite{benson16} in order to explain what it means to use motif clustering based on a collection of motifs\footnote{This is also how functional motifs can be included in the formalism.}. 

Given motifs $M_1, \ldots, M_q$, coefficients $\alpha_1, \ldots, \alpha_q \in \R$ such that $\alpha_j > 0$ for all $j\in [q]$, and a vertex subset $W \subset V$, we can consider weighted motif cuts
\[
    \text{cut}_{(G, \{\alpha_j, M_j\}_{j=1}^q)} (W) := \sum_{j=1}^q \alpha_j \text{cut}_{(G, M_j)} (W)\,,
\]
and the weighted motif volume
\[
    \text{vol}_{(G, \{\alpha_j, M_j\}_{j=1}^q)} (W) := \sum_{j=1}^q \alpha_j \text{vol}_{(G, M_j)} (W) \,.
\]
We can also define the corresponding weighted motif conductance 
\[
    \phi_{(G, \{\alpha_j, M_j\}_{j=1}^q)}(W) := \frac{\text{cut}_{(G, \{\alpha_j, M_j\}_{j=1}^q)} (W)}{\text{vol}_{(G, \{\alpha_j, M_j\}_{j=1}^q)} (W)}\,.
\]
and weighted motif ratio cut 
\[
    \text{RatioCut}_{(G, \{\alpha_j, M_j\}_{j=1}^q)}(W) := \frac{\text{cut}_{(G, \{\alpha_j, M_j\}_{j=1}^q)} (W)}{|W|}\,,
\]
which naturally extend to partitions $W_1, \ldots, W_k$ of $V$ as follows:
\[
    \phi_{(G, \{\alpha_j, M_j\}_{j=1}^q)}(W_1, \ldots, W_k) := \sum_{i=1}^k \phi_{(G, \{\alpha_j, M_j\}_{j=1}^q)}(W_i)\,,
\]
for weighted motif conductance, and 
\[
    \text{RatioCut}_{(G, \{\alpha_j, M_j\}_{j=1}^q)}(W_1, \ldots, W_k) := \sum_{i=1}^k \text{RatioCut}_{(G, \{\alpha_j, M_j\}_{j=1}^q)}(W_i)\,,
\]
for weighted motif ratio cut.

\

\noindent Now, if we are interested in finding partitions of the vertex set $V$ of $G$ that approximately minimize the weighted motif conductance or weighted motif ratio cut in $G$ with respect to weights $\{\alpha_j\}_{j=1}^q$ and motifs $\{M_j\}_{j=1}^q$ then, as before, we can instead minimize ordinary conductance or ordinary ratio cut of the weighted motif graph defined below, \emph{as long as} the motifs $M_j$ either all have two anchor nodes -- the case we will use -- or all have three anchor nodes. 

In order to construct the weighted motif graph, we construct the motif graph $\fG_j$ with motif adjacency matrix $\fA_j$ for every $j\in [q]$, and then take a weighted linear combination ($\alpha_j \geq 0$ for $j \in [q]$)
\[
    \fA_\Sigma:= \sum_{j=1}^q \alpha_j \fA_j
\]
to construct a single weighted motif graph $\fG_\Sigma$ that combines all motifs according to the weights $\alpha_j$ (which determine the relative importance of each motif $M_j$). Writing $N_A$ for the number of anchor nodes that all motifs have (recall that they should all have the same amount of anchor nodes, either two or three), then weighted motif cuts in the original graph can be directly related to ordinary cuts in the weighted motif graph. Specifically, given a subset $W \subset V$ of the vertex set $V$ of $G$, we have from Lemma~\ref{lem:benson1} that
\begin{equation}
    \sum_{j=1}^q \alpha_j \text{cut}_{(G, M_j)} (W) = \sum_{j=1}^q \alpha_j \, c \, \text{cut}_{\fG_j}(W) = c \, \text{cut}_{\fG_\Sigma}(W)\,,
    \label{eq:weigted_motif_cut}
\end{equation}
where $c=1$ if $N_A = 2$, and $c=\frac{1}{2}$ if $N_A = 3$, and
\begin{equation}
    \sum_{j=1}^q \alpha_j \text{vol}_{(G, M_j)} (W) = \sum_{j=1}^q \frac{\alpha_j}{N_A -1} \text{vol}_{\fG_j}(W) = \frac{1}{N_A - 1} \text{vol}_{\fG_\Sigma}(W)\,.
    \label{eq:weighted_motif_vol}
\end{equation}
Consequently, we can apply apply spectral clustering to the normalized or unnormalized Laplacian of $\fG_\Sigma$, respectively, to obtain a partition of $V$ that approximately minimizes either the weighted motif conductance or ratio cut.

\subsection{Motifs with more than two anchor nodes}
\label{sec:motifs_more_than_two}

Having introduced weighted linear combinations of motifs, we next turn to motifs with more than two anchor nodes. Let $G$ be a graph, $M = (V_M, E_M, V_A)$ be a motif with more than $|V_A| > 2$ anchor nodes, and let $\fG_M$ be the corresponding motif graph\footnote{We add the $M$-dependency to the notation explicitly in this subsection.} with adjacency matrix $\fA_M$. In this section we will show that the motif graph $\fG_M$ constructed from $G$ and $M$ is equal to the motif graph obtained by taking a weighted sum of all two-anchor-node motifs contained in $M$, with the weights chosen as described below. Subsequently, we will use this result to show that motif clustering using a three-anchor-node motif is equivalent to clustering using a combination of two-anchor-node motifs.

\subsubsection{The motif graph for motifs with more than two anchor nodes}
\label{sec:motif_graph_more_than_two}

Let $A_2(M) = \{M_j\}_{j=1}^q$ be the set of all two-anchor motifs with the same vertex and edge sets as $M$ such that the anchor node sets of the motifs $\{M_j\}_{j=1}^q$ enumerate all possible two-element subsets of $V_A$, modulo motif isomorphisms. 

For example, if $M = (V_M, E_M, V_A)$ has three anchor nodes $V_A = \{a,b,c\}$, then in the absence of motif isomorphisms $A_2(M)=\{M_1, M_2, M_3\}$ consists of three motifs given by $M_1 = (V_M, E_M, \{a,b\})$, $M_2 = (V_M, E_M, \{b,c\})$ and $M_3 = (V_M, E_M, \{c,a\})$. However, if it so happens to some of the motifs in $A_2(M)$ are motif-isomorphic, then we only keep one motif per equivalence class. E.g.~for the three-anchor-node motif in the top left of Fig.~\eqref{fig:anchor_node_symmetry}, we observe that $M_2$ and $M_3$ are motif isomorphic. Hence, $A_2(M) = \{M_1, M_2\}$ or $A_2(M) = \{M_1, M_3\}$ (it does not matter which of the two we pick --- see below).
\begin{figure}[!h]
    \centering
    \includegraphics[scale=1]{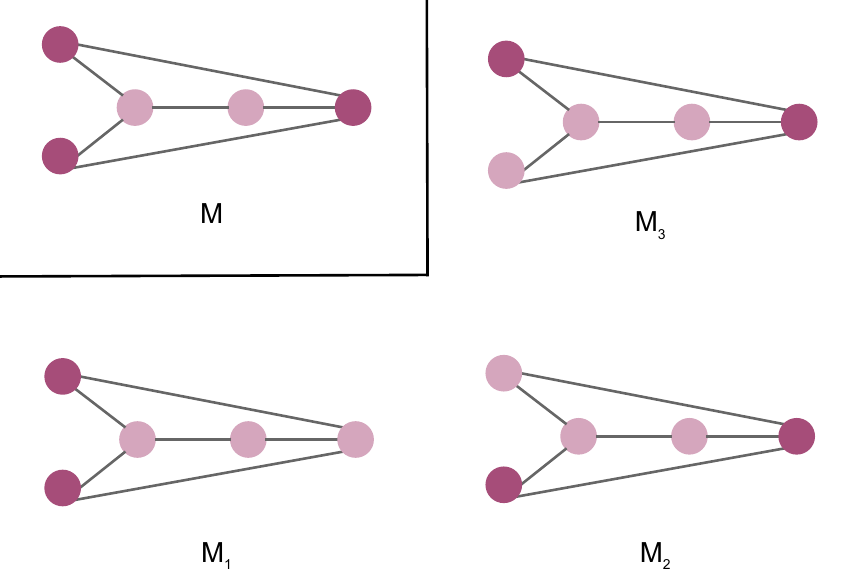}
    \caption{Top left: three-anchor-node motif with a symmetry $M$. For clustering we need $M_1$ and either $M_2$ or $M_3$ (because the latter two motifs give rise to the same motif instances in the graph $G$).}
    \label{fig:anchor_node_symmetry}
\end{figure}

Next, for every two-anchor node motif $K \in A_2(M)$, we define the weight $\omega_K$ as follows. Let $u,v \in V_M$ be the two anchor nodes of $K$. Now, $\omega_K$ is given by the number of ways the remaining $|V_A| - 2$ anchor nodes (i.e~all anchor nodes except $\{u,v\}$) can be assigned to the graph $(V_M, E_M)$ such that the motif structure of $M$ is respected. In other words: $\omega_K$ is the number of motif instances $\iota$ of $M$ in the graph $(V_M, E_M)$ for which  $\{u,v\} \in \iota(V_A)$. An example of a motif and the corresponding two-anchor-node weights is given in Fig.~(\ref{fig:three-anchor-node-symmetry}).

\begin{figure}[!h]
    \centering
    \includegraphics[scale=0.5]{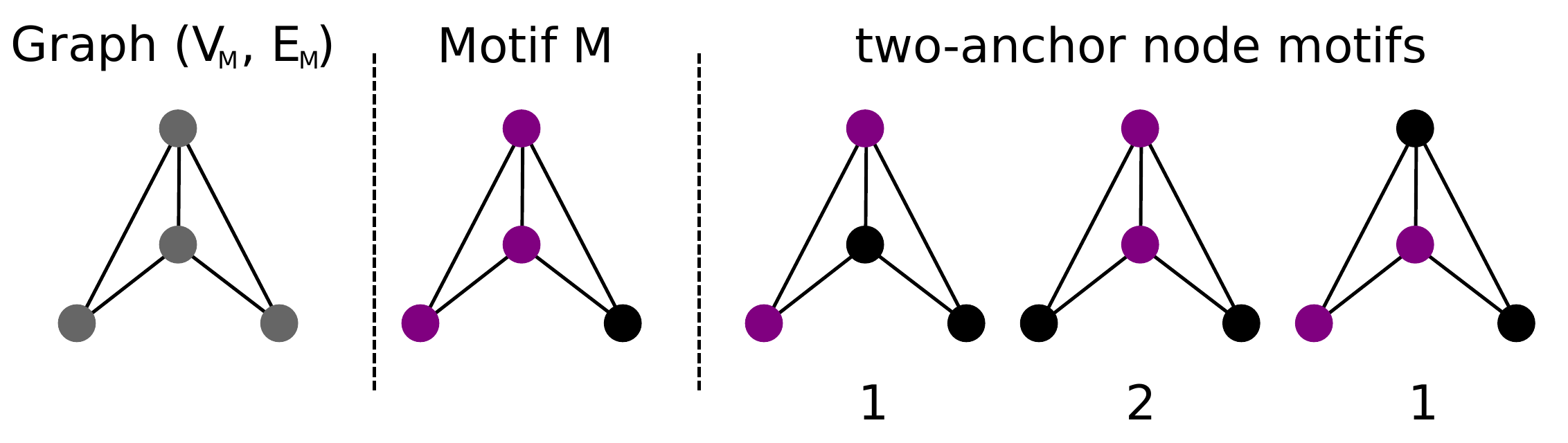}
    \caption{Left: graph $(E_M, V_M)$ and the  three-anchor-node motif $M$. Right: all three two-anchor-node motifs in $A_2(M)$, with their corresponding weights displayed by the integers below. Indeed, for the middle two-anchor-node motif, the remaining anchor node (bottom left in $M$) can be mapped to either two bottom nodes of $(V_M, E_M)$, as both mappings respect the motif structure, so the weight is 2. For the other two-anchor-node motifs, there is only one vertex that the remaining anchor node can be mapped to such that the motif structure is preserved.}
    \label{fig:three-anchor-node-symmetry}
\end{figure}

Now, define the weighted motif graph $\fG_{A_2(M)}$ obtained by taking the weighted sum of all two-anchor-node motifs in $K \in A_2(M)$ weighted by $\omega_K$. Its adjacency matrix is given by
\[
    \fA_{A_2(M)} = \sum_{K \in A_2(M)} \omega_K  \fA_K,
\]
where for each $K \in A_2(M)$, $\fA_K$ is the adjacency matrix of the motif graph obtained from $G$ and the two-anchor-node motif $K$. 

Note that $\fA_{A_2(M)}$ is independent of the choice of representative from each motif isomorphism class\footnote{By symmetry, different representative have the same weight.}. In the example of Figure~\ref{fig:anchor_node_symmetry}, the adjacency matrix $\fA_{M_2}$ obtained by finding all instances of $M_2$ in $G$ is the exact same matrix as $\fA_{M_3}$ obtained by finding all instances of $M_3$ in $G$, and therefore (in this case that all weights are equal to one), $\fA_{A_2(M)} = \fA_{M_1} + \fA_{M_2} = \fA_{M_1} + \fA_{M_3}$.

\begin{lemma}
For a graph $G$ and motif $M = (V_M, E_M, V_A)$, the motif adjacency matrices $\fA_M$ and $ \fA_{A_2(M)}$ as defined above, and hence also the motif graphs $\fG_M$ and $\fG_{A_2(M)}$, are equal.
\end{lemma}

\begin{proof}
Let $s = |V_M|$ be the number of vertices of the motif $M$. In order to prove that $\fA_M = \fA_{A_2(M)}$, it suffices to check that each $s$-sized set of distinct vertices of $G$ contributes equally to both $\fA_M$ and $\fA_{A_2(M)}$. Hence, choose $s$ distinct vertices $V' = \{v_1, \ldots, v_s\}$ of $G$, and let $G'$ be the graph $G$ restricted to $V'$. In order to compute what $G'$ contributes to $\fA_M$, we need to (1) find all motif instances of $M$ in $G'$, and (2) for each motif instance, add $+1$ to the weight of each edge connecting two anchor nodes of the motif instance in question. We then do the same for all motifs that make up $\fA_{A_2(M)}$, and check that the two contributions are equal.

There are two options, either there is a vertex assignment $\iota : V_M \rightarrow V'$ such that $\iota$ is a graph isomorphism from $(V_M,E_M)$ to $G'$, or no such assignment exists. In the latter case, $G'$ contributes nothing, hence equally, to both $\fA_M$ and $\fA_{A_2(M)}$, since all motifs in $A_2(M)$ have the same edge pattern as $M$ does. Thus, we need to only focus on the former case, for which $G'$ is graph-isomorphic to $(V_M, E_M)$, and we can identify $V' \simeq V_M$.

Assuming $G' \simeq (V_M, E_M)$, we now need to show that for every vertex pair $u,v \in V'$, (i) the contribution of all motif instances of $M$ in $G'$ to $(\fA_{M})_{uv}$ is equal to (ii) the sum of contributions of motif instances in $G'$ of all motifs $K\in A_2(M)$ to $(\fA_{K})_{uv}$ weighted by their weights $\omega_K$. Hence, pick a pair $u,v \in V' \simeq V_M$. If no motif instance of $M$ exists for which $u$ and $v$ are anchor nodes, then neither will a motif instance of any $K\in A_2(M)$ exist that has $u$ and $v$ as anchor nodes, and the contributions (i) and (ii) will both be zero, hence equal.


If, on the other hand a motif instance of $M$ exists for which $u$ and $v$ are anchor nodes, then (i) is given by the number of motif instances $\iota$ in $G' \simeq (V_M, E_M)$ for which both $u$ and $v$ are anchor nodes (i.e.~$u,v \in \iota(V_A)$). By definition, this number is exactly equal to $\omega_{\tilde{K}}$, where $\tilde{K}$ is the two-anchor-node motif $\tilde{K} = (V_M, E_M, \{u,v\})$. Because $A_2(M)$ contains one motif of each motif-isomorphism class of two-anchor-node motifs of $M$, there is exactly one motif in $A_2(M)$ -- either $\tilde{K}$ itself, or one of the isomorphic two-anchor node motifs that we pick as representative -- that has a motif instance in $(V_M, E_M)$ of which $u$ and $v$ are anchor nodes, and this motif has weight $\omega_{\tilde{K}}$. Therefore $G'$ contributes the same, namely $\omega_{\tilde{K}}$, to both $(\fA_{M})_{uv}$ and $(\fA_{A_2(M)})_{uv} = \sum_{K\in A_2(M)} \omega_K (\fA_{K})_{uv}$.


Since this analysis holds for every vertex pair $u,v \in V'$, we conclude that $G'$ contributes equally to $\fA_M$ and $\fA_{A_2(M)}$. Because the above analysis holds for every $s$-sized set of distinct vertices of $G$, we conclude that the adjacency matrices $\fA_M$ and $\fA_{A_2(M)}$ are equal, and therefore $\fG_M = \fG_{A_2(M)}$. 
\end{proof}
\noindent As a consequence, if we want to construct the motif graph $\fG_M$ for any given motif $M$, we can instead construct the motif graphs $\fA_K$ for every $K \in A_2(M)$, and then sum the resulting motif adjacency matrices weighted with $\omega_K$ to obtain $\fA_M$. 

\

\noindent We can also do the above approximately. That is, for some fixed $\epsilon >0$, if we have for every $K \in A_2(M)$ an approximation $\tilde{\fA}_K$ of $\fA_K$ such that for every $u,v\in V$
\[
    (1-\epsilon) (\fA_K)_{uv} \leq (\tilde{\fA}_K)_{uv} \leq (1+\epsilon) (\fA_K)_{uv}\,,
\]
then by summing over $K \in A_2(M)$ weighted with $\omega_K$, we will obtain an approximation $\tilde{\fA}_M = \sum_{K\in A_2(M)} \omega_K \tilde{\fA}_K$ of  $\fA_M = \sum_{K\in A_2(M)} \omega_K \fA_K$ up to relative error $\epsilon$ for each $u,v\in V$:
\[
    (1-\epsilon) (\fA_M)_{uv} \leq (\tilde{\fA}_M)_{uv} \leq (1+\epsilon) (\fA_M)_{uv}.
\]
In particular, if each $\tilde{\fA}_K$ were constructed via quantum approximate counting with relative error $\epsilon$, then the sum $\tilde{\fA}_M = \sum_{K\in A_2(M)} \omega_K \tilde{\fA}_K$ will approximate  (coefficient-wise) the motif adjacency matrix $\fA_M$ also up to relative error $\epsilon$.\footnote{If we use this in the context of Section~\ref{sec:quantum_algorithms}, then to ensure that the entire procedure succeeds with probability at least $1-\delta$, for given $K$ and given $u,v \in V$ we need to run Algorithm~\ref{algo:motif_count} with success probability $1- \delta / (|A_2(M)|n^2)$ (rather than $1- \delta / n^2$) in line 7 of Algorithm~\ref{algo:motif_count}, which yields an extra factor of $\log(|A_2(M)|)$ to both the query count and the number of other operations in Lemma~\ref{lem:approx_adjacencymatrix}.}

\subsubsection{Three-anchor-node motifs as a combination of two-anchor-node motifs}
\label{sec:three_as_two}

Using the results from the previous subsection, we can reduce the case of clustering with a three-anchor-node motif to that of clustering using three separate two-anchor-node motifs.

Let $G$ be a graph, $M$ a three-anchor-node motif with corresponding motif graph $\fG_M$, and let $A_2(M)$ and $\fG_{A_2(M)}$ be as in the subsection above. Now, For any subset $W \subset V$, we have
\begin{align*}
    \text{cut}_{(G,M)}(W) &= \frac{1}{2} \text{cut}_{\fG_M}(W) \\
    &= \frac{1}{2}\text{cut}_{\fG_{A_2(M)}}(W) \\
    &= \frac{1}{2} \sum_{K \in A_2(M)} \omega_K \text{cut}_{\fG_K} (W) \\
    &= \frac{1}{2} \sum_{K \in A_2(M)} \omega_K \text{cut}_{(G,K)} (W) \numberthis \label{eq:cut-two-three-anchors}
\end{align*}
where in the first line we used Lemma~\ref{lem:benson1} and the fact that $M$ has three anchor nodes (hence the factor $\frac{1}{2}$), in the second we use the fact that $\fG_M = \fG_{A_2(M)}$, in the third we use that $\fG_{A_2(M)}$ is a weighted sum of all two-anchor-node motifs $K \in A_2(M)$, and in the last line we again use Lemma~\ref{lem:benson1} for each two-anchor-node motif $K$.


Fig~\ref{fig:cut_three_anchor_nodes} shows an example of how a cut in a three-anchor-node motif $M$ in the original graph (left motif of the figure) cuts two out of the three motifs $M_1$, $M_2$ and $M_3$ (right three motifs in figure), visualizing the factor $\frac{1}{2}$ in Eq.~\eqref{eq:cut-two-three-anchors} (note that all weights are one in this case). 
\begin{figure}[!h]
    \centering
    \includegraphics[scale=1]{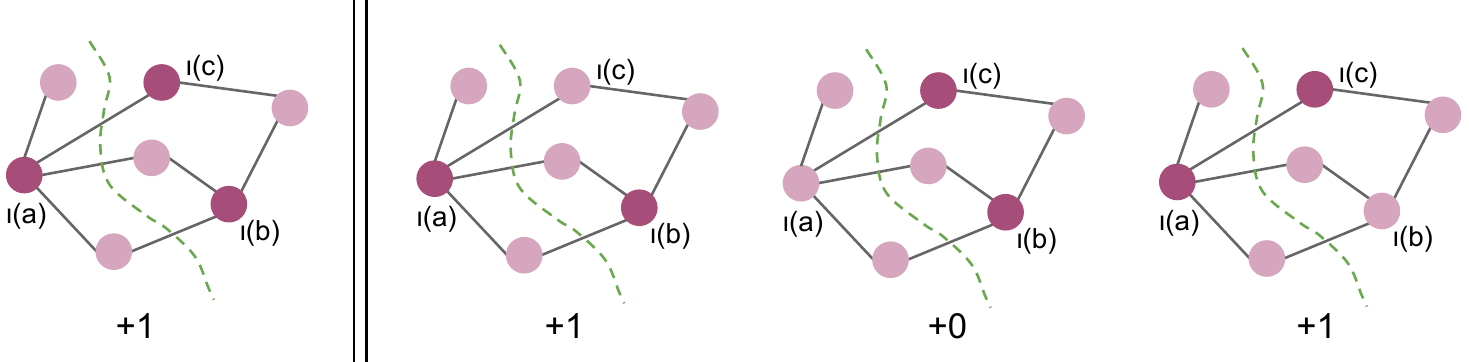}
    \caption{\emph{Left} --- A motif instance of $M$ with three anchor nodes (shown in purple) cut into two: so this motif contributes $+1$ to the left hand side of Eq.~\eqref{eq:cut-two-three-anchors}. \emph{Right} --- the same cut separates anchor nodes for two out of the three motif instances of $M_1, M_2, M_3$ corresponding to the motif instance of $M$ on the left, and therefore also contributes $+1$ to the right hand side of Eq.~\eqref{eq:cut-two-three-anchors}.}
    \label{fig:cut_three_anchor_nodes}
\end{figure}

Similar to Eq~\eqref{eq:cut-two-three-anchors}, and using the exact same derivation but now with `$\text{vol}$' replaced by `$\text{cut}$', we find that the motif volumes are related by:
\begin{equation}
     \text{vol}_{(G,M)}(W) =  \frac{1}{2} \sum_{K \in A_2(M)} \omega_K \text{vol}_{(G,K)} (W)\, .
     \label{eq:vol-two-three-anchors}
\end{equation}
Consequently, we have
\[
    \phi_{(G,M)}(W) = \phi_{(G, A_2(M))}(W)
\]
and
\[
    \text{RatioCut}_{(G,M)}(W) = \frac{1}{2}\text{RatioCut}_{(G, A_2(M))}(W)\, .
\]
Where in the notation for conductance and RatioCut above '$A_2(M)$' implies that the weights $\{\omega_K\}_{K\in A_2(M)}$ are used to compute the weighted cuts and volumes.

Because the above holds for every subset $W\subset V$, it also holds for partitions $\{W_i\}_{i=1}^k \in \mP_k(V)$. Hence, we also have that 
\[
    \argmin_{\{W_i\}_{i=1}^k \in \mP_k(V)}\phi_{(G,M)}(W_1, \ldots, W_k) = 
    \argmin_{\{W_i\}_{i=1}^k \in \mP_k(V)}\phi_{(G, A_2(M))}(W_1, \ldots, W_k) 
\]
and
\[
    \argmin_{\{W_i\}_{i=1}^k \in \mP_k(V)}\text{RatioCut}_{(G,M)}(W_1, \ldots, W_k) = 
    \argmin_{\{W_i\}_{i=1}^k \in \mP_k(V)}\text{RatioCut}_{(G, A_2(M))}(W_1, \ldots, W_k). 
\]

In conclusion, performing motif clustering on any three-anchor-node motif $M$ is equivalent to performing motif clustering on $\omega_K$-weighted combination of all motifs in $A_2(M)$ obtained by considering all possible pairs of anchor nodes of $M$, modulo motif isomorphisms.

\subsection{Motifs with more than three anchor nodes, and hypergraphs}
\label{sec:hypergraphs}

For motifs $M$ with more than three anchor nodes, the motif graph $\fG_M$ is still equal to the motif graph $\fG_{A_2(M)}$ corresponding to the $\omega_K$-weighted sum of all two-anchor-node motifs in $A_2(M)$. As a consequence, we can apply spectral clustering to the motif graph $\fG_{A_2(M)}$ to obtain a $k$-partition $\{W_1, \ldots, W_k\}$ of the vertex set $V$ that approximately minimizes $\text{RatioCut}_{(G, A_2(M))}(W_1, \ldots, W_k)$ or $\phi_{(G, A_2(M))}(W_1, \ldots, W_k)$ of the $\omega_K$-weighted sum of all motifs in $A_2(M)$.

However, what no longer holds is that the above RatioCut and conductance are proportional to the RatioCut and conductance corresponding to the motif $M$. In other words: if $M = (V_M, E_M, V_A)$ is a motif with $|V_A| > 3$, then there \emph{does not} exist in general a constant $c>0$ such that for all subsets $W \subset V$ we have that $\phi_{(G,M)}(W) = c \, \phi_{(G, A_2(M))}(W)$, \emph{nor} does there exist a constant $c'$ such that $\text{RatioCut}_{(G,M)}(W) = c' \, \text{RatioCut}_{(G, A_2(M))}(W)$. The reason is that, for a given motif instance, the factor of proportionality depends on how exactly the motif instance is cut. 

For example, if the motif has four anchor nodes, then we see in Fig~\ref{fig:four_anchor_node_cut} that a cut in $G$ that separates one anchor node from the rest in a given motif instance adds $+3$ to the corresponding cut in $\fG_M$, but a cut through the middle that separates two anchor nodes from the rest adds $+4$ to the corresponding cut in $\fG_M$. In contrast, both cuts contribute $+1$ to $\text{cut}_{(G,M)}(W)$ in the original graph $G$.

\begin{figure}[!h]
    \centering
    \includegraphics{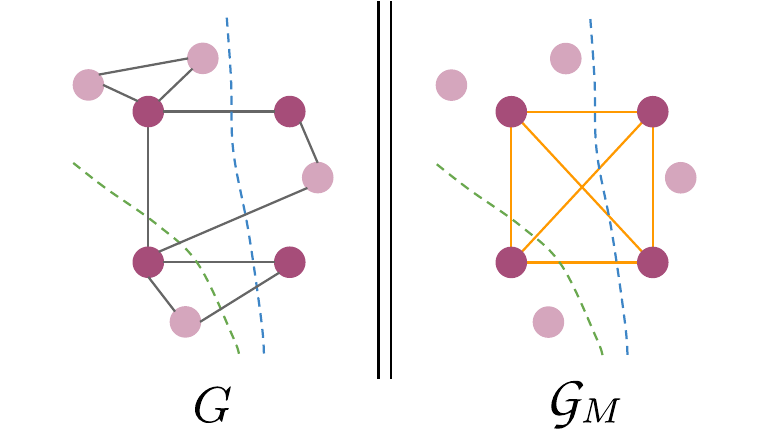}
    \caption{Left --- a four-anchor node motif instance $\iota$ in the input graph $G$. For both the green and the blue cut, $\iota$ adds $+1$ to the motif cut in the original graph $G$, because bots cuts have anchor nodes of $\iota$ to the left and right of the cut. Right --- corresponding motif graph $\fG_M$, which has a four-clique connecting the four anchor nodes coming from $\iota$. The green line cuts three edges, so $\iota$ adds $+3$ to the green cut in $\fG_M$, whereas the blue line cuts through four lines, so here $\iota$ adds $+4$ to the cut.}
    \label{fig:four_anchor_node_cut}
\end{figure}

The above discrepancy is the reason why, in Theorem~9 of their supplementary material, Benson et al.~\cite{benson16} subtract the sum of all motif instances for which the anchor node set is cut exactly in half by the cut in the graph. Note that, for a motif with three anchor nodes, any cut separates one anchor node from two other ones, and therefore the issue that arises with motifs with four anchor nodes is not present there.

\

\noindent To conclude, there are two ways of viewing a motif: one is that the motif expresses a single multiple-vertex relationship between all of its anchor nodes; the other that the motif represents the combination of all pairwise relationships between its anchor nodes. Because the motif graph can only represent pairwise relationships, clustering using the motif graph only works if we take the latter viewpoint.

More explicitly, for motifs with two anchor nodes, and also for arbitrary weighted combinations of motifs with two anchor nodes, we can do motif clustering by means of the motif graph. However, for any motif $M$ with more than two anchor nodes, clustering using the motif graph $\fG_M$ should \emph{only} be used to find a $k$-partition $\{W_1, \ldots, W_k\}$ of $V$ that approximately minimizes $\phi_{(G, A_2(M))}(W_1, \ldots, W_k)$ or $\text{RatioCut}_{(G, A_2(M))}(W_1, \ldots, W_k)$. Such a clustering would take the view that a motif represents a collection of pairwise relationships between its anchor nodes. 

Coincidentally, for three-anchor-node motifs, clustering using the $\omega_K$-weighted sum of motifs in $A_2(M)$ is equivalent to clustering using $M$ directly, in which case the motif graph $\fG_M$ \emph{can} also be used to find a $k$-partition that approximately minimizes $\phi_{(G, M)} (W_1, \ldots, W_k)$ or $\text{RatioCut}_{(G,M)}(W_1, \ldots, W_k)$.

If, instead, we take the viewpoint above that a motif represents a multiple-vertex relationship, then we are interested in finding a $k$-partition $\{W_1, \ldots, W_k\}$ of $V$ that approximately minimizes $\text{RatioCut}_{(G,M)}(W_1, \ldots, W_k)$ or $\phi_{(G,M)}(W_1, \ldots, W_k)$. For a motif $M$ with more than three anchor nodes we should \emph{not} use the motif graph $\fG_M$ to do this (since $\fG_M$ only expresses pairwise relationships between vertices). Rather, if we truly want to study higher-order relationships, we should construct a hypergraph that expresses the multiple-vertex relationships and cluster the vertices of this hypergraph. For this task there exist classical algorithms: see for example~\cite{chan2018spectral} for spectral properties of the hypergraph Laplacian, and~\cite{takai2020hypergraph} and references therein for classical algorithms for clustering hypergraphs.

\bibliographystyle{alphaurl}
\bibliography{main.bib}

\end{document}